\documentclass[12pt]{iopart}
\usepackage{iopams}
\usepackage{hyperref}
\usepackage{mathrsfs}
\usepackage{braket}
\usepackage{graphicx}
\usepackage{subcaption}
\usepackage{xcolor}
\usepackage{soul}
\usepackage{lineno}
\newtheorem{theorem}{Theorem}[section]
\newtheorem{definition}{Definition}[section]
\newtheorem{lemma}{Lemma}[section]
\newenvironment{proof}{\paragraph{{\bf Proof:}}}{\hfill$\square$}
\newcommand{\binom}[2]{\left(\!\!\begin{array}{c}{#1}\\{#2}\end{array}\!\!\right)}
\newcommand{\abs}[1]{\left|{#1}\right|}
\newcommand{\norm}[1]{\left\|{#1}\right\|}
\def\eqdef{\stackrel{\mathrm{def}}{=}}
\def\bR{\mathbb{R}}
\def\bN{\mathbb{N}}
\def\supp{\mathrm{supp}}
\def\fob{f_{\mathrm{o}}}
\def\fballs{f_{\mathrm{b}}}
\def\rmv{\mathrm{v}}
\def\rmh{\mathrm{h}}
\def\nv{n_{\rmv}}
\def\nh{n_{\rmh}}
\def\one{1}
\def\groupeG{\mathrm{g}}
\def\slalphaihat{\hat{s}^l_{\alpha_i}}
\def\salphaihat{\hat{s}_{\alpha_i}}
\def\tildev{\tilde{v}}
\def\alphahat{\hat{\alpha}}
\def\alphahatz{\alphahat_0}
\def\alphahati{\alphahat_i}
\def\ahat{\hat{a}}
\def\nbk{n_k}
\def\ErrS{\mbox{ErrS}}
\def\ErrAone{\mbox{ErrA}^{I}}
\def\ErrAtwo{\mbox{ErrA}^{II}}
\def\ErrLambda{\mbox{Err}\Lambda}
\def\ErrY{\mbox{ErrY}}
\def\ErrC{\mbox{ErrC}}
\def\ErrP{\mbox{ErrP}}
\def\alphahatone{\hat{\alpha}^{I}}
\def\alphaihatone{\alphahatone_i}
\def\alphazhatone{\alphahatone_0}
\def\alphahattwo{\hat{\alpha}^{II}}
\def\alphaihattwo{\alphahattwo_i}
\def\alphazhattwo{\alphahattwo_0}
\def\deltacMt{\left(\delta_{\vec{c}}\right)_{M,\vec{t}}}
\def\lone{\rm{a}}
\def\ltwo{\rm{b}}
\def\ftol{\left(\sum_{j=1}^{4} \delta_{\vec{c}^{\:l}_j}\right)}
\newcommand{\Kone}[1]{K_{1,#1}}
\def\Konej{\Kone{j}}
\def\Konejpo{\Kone{j+1}}
\def\KoneJ{\Kone{J}}
\def\Konejbbeta{K_{1,j(\bbeta)}}
\def\KonejbbetaI{K_{1,j(\bbeta)}}
\newcommand{\Ktwo}[1]{K_{2,#1}}
\def\Ktwoi{\Ktwo{i}}
\def\Ktwoipo{\Ktwo{i+1}}
\def\KtwoI{\Ktwo{I}}
\newcommand{\Pone}[1]{P_{1,#1}}
\def\Ponej{\Pone{j}}
\def\PoneJ{\Pone{J}}
\newcommand{\Ptwo}[1]{P_{2,#1}}
\def\Ptwoi{\Ptwo{i}}

\def\PtwoI{\Ptwo{I}}
\def\gammaBd{\gamma}
\def\GammaBB{\Gamma(\bbeta)}
\def\Ponetjgamma{\tilde{P}_{1,j,\gamma}}
\def\PonetJgamma{\tilde{P}_{1,j(\bbeta),\gammaBd}}
\def\PonetJgammabeta{\tilde{P}_{1,j(\bbeta),\gamma(\bbeta)}}
\def\Ptwotibbeta{\tilde{P}_{2,i,\bbeta}}
\def\PtwotIbbeta{\tilde{P}_{2,I,\bbeta}}

\newcommand{\redC}[1]{\textcolor[rgb]{0.00,0.00,0.00}{#1}}

\begin{document}
\title[\small Range conditions on distributions  and their possible application to geometric calibration]{\redC{Range conditions on distributions  and their possible application to geometric calibration in 2D parallel and fan-beam geometries}} 

\author{Anastasia Konik$^{1}$ and Laurent Desbat$^1$}
\address{$^1$TIMC-IMAG, Univ. Grenoble Alpes, 38700, La Tronche, France}

\begin{abstract}
\redC{In tomography, range conditions or data consistency conditions (DCCs) on functions have proven useful for geometric self-calibration, which involves identifying geometric parameters of acquisition systems based only on acquired radiographic images. These self-calibration methods using range conditions on functions typically require non-truncated data. In this work, we derive range conditions on distributions and demonstrate their application in addressing data truncation issues during the calibration process. We propose a novel approach based on range conditions on distributions, employing Dirac distributions to model markers within the field-of-view of an X-ray system.  Our calibration methods are based on the local geometric information from non-truncated projections of a marker set.  By applying range conditions to projections of sums of Dirac distributions, combined with specific calibration marker sets, we derive analytical formulas that enable the identification of geometric calibration parameters. We aim to present DCCs on distributions in tomography and explore the potential of DCCs on distributions as a possible tool in calibration. This approach represents one possible application, demonstrating how DCCs on distributions can effectively address challenges such as data truncation and incomplete marker set information. We present results for the 2D parallel geometry (Radon transform) and the 2D fan-beam geometry with sources on a line.}
 
\end{abstract}

\vspace{2pc}

\section{Introduction}

In tomography, X-ray projection data are acquired in order to reconstruct the  attenuation function of a measured object (in industrial applications) or patient organs (in medical applications). After physical corrections (gain, offset, log, etc.), the projection data can be modeled by integrals of the attenuation  function over the X-ray lines, see~\cite{Natterer:86} and the subsection~\ref{Sec2DRadon}.
To reconstruct the attenuation function by solving the corresponding linear inverse problem, the acquisition system (e.g., the scanner in computed tomography (CT)) must be perfectly known. In particular, a geometric calibration must be performed to precisely model the acquisition geometry.

Numerous authors have shown that images reconstructed from the data acquired from acquisition systems suffering from noncorrected geometric misalignments may have very strong artifacts~\cite{Hsieh99,VidalMigallon2008,panetta08,Kyriakou_2008,Kingston11,Ouadah16a}.
Many X-ray systems are geometrically calibrated "offline" (before their routine use), because their mechanics are robust, precise and reproducible. In this case, dedicated geometric calibration systems are used, very often including opaque markers, and the calibration is performed at a relatively low frequency, for example, for a cone-beam (CB) system, see~\cite{Rougee1993}. Geometric calibration methods of 3D CB systems are often derived from computer vision~\cite{Hartley2000,Menessier09}. Some geometric calibration methods use opaque markers of unknown position with strong {\em a priori} information (such as the source trajectory is circular) to identify analytically a small number of geometric parameters~\cite{NooClack00,Smekal2004}. Mathematical properties of the acquisition geometry (symmetries) have been used for the calibration of a turn table~\cite{Patel2009,Meng13} or an acquisition system with limited number of parameters~\cite{panetta08} without any use of calibration markers was calibrated with iterative methods for minimizing a cost function. 

The development of mobile C-arm in the early 2000s led to the design of "online" calibration methods, i.e., a geometric calibration of each acquisition. Mobile X-ray system mechanics are generally not sufficiently robust, precise and reproducible for 2D or 3D image reconstructions from projections with "offline" geometric calibration only. 
The geometric calibration is then performed during the projection acquisition.  It's called self-calibration when projection data only are used. Many self-calibration methods require a non-quadratic cost function to be minimized. Very often the number of geometric calibration parameters to be estimated must be reduced to the most sensitive ones in order to give a chance to the optimization method to compute a solution. 
In~\cite{Kyriakou_2008} the histogram entropy of the reconstructed image is minimized according to geometric calibration parameters. In~\cite{Kingston11} the cost is based on the sharpness of the reconstruction and only the 4 most sensitive geometric calibration parameters of a turn table are estimated. Some methods require a pre-existing 3D image (of the same patient) to be registered with the acquired 2D CB projection through the CB X-ray projection simulation (called Digitally Reconstructed Radiograph) and an iterative process~\cite{Ouadah16a,Otake_2013}.
These geometric self-calibration methods have the great advantage of not using opaque calibration markers perturbing the X-ray acquisition and the image reconstruction. Calibration marker geometries exploiting projection data redundancy have been designed for the 3D parallel geometry, see~\cite{desbat02b}, in order to reduce the impact of markers on the reconstruction, but at the very high cost of moving markers during the projections.

In tomography, self-calibration methods based on range conditions (also called data consistency conditions (DCCs), i.e., mathematical properties of projection data expressing redundancy) have been proposed. The first DCCs for CB CT are probably due to John~\cite{JohnFritz1938Tude}. From John's equations, DCCs have been established for X-ray sources belonging to a plane~\cite{FinchDavid1985CBRw} (see also~\cite{ClackdoyleRolf2013Fdcc} with further arguments, see~\cite{LevineMargoS2010CCfC} for sources on a line and~\cite{Nguyen_2024} for links between different DCCs associated with different CB geometries with co-planar sources) or on a helical trajectory~\cite{Patch2002Cout}. DCCs have also been constructed for sources on a sphere~\cite{finch1983}. 
In parallel geometry, DCCs known as Helgason-Ludwig consistency conditions (HLCCs)~\cite{helgason65,ludwig66,gelfand66} are based on projection moments. 
In 2D parallel geometry, Basu and Bressler studied and provided sufficient conditions for using the HLCCs of the 2D Radon transform to perform a geometric self-calibration. They have shown that the projection angles and detector shifts can be estimated from moments of order 0, 1, 2 and 3 only of sufficiently many projections with iterative methods. They also have shown that the data are consistent with rotations, translations and symmetries of the measured object resulting in a constant angular shift, a sinusoidal detector shift and a symmetry on the projection angle, see~\cite{basu100,basu200}.
Several self-calibration methods based on range conditions often require iterative optimization techniques to solve non-linear problems to identify the geometric calibration parameters~\cite{lesaint2017a, aichert15}. Range conditions for non-truncated fan-beam data with sources on a line were proposed in~\cite{clackdoyle13} and the corresponding self-calibration algorithm with a closed-form solution in~\cite{nguyen20}. Note that in~\cite{desbat14} a similar approach is proposed with a  closed-form solution for the self-calibration of detector shifts in the 2D parallel tomography.

One drawback of moment methods, such as HLCCs, is that the projections must not be truncated. The projection truncation prevents the computation of projection moments.  In many situations, X-ray projections are truncated due to the small detector size or in order to avoid unnecessary irradiation. Range conditions for truncated data in 2D parallel and fan-beam with sources on a circle geometries were introduced in~\cite{clackdoyle15} and in fan-beam geometry with an arbitrary source trajectory in~\cite{yu15}. However, these approaches imply data re-binning (or equivalently, a change of variables in moment integrals). This mixes the geometric calibration parameters. This is a huge barrier to find an efficient closed-form solution to the geometric calibration problem.

In computer vision, some geometric self-calibration methods of multiple views are based on singularities detected in images~\cite{Hartley2000}.  Bundle adjustment (BA) methods have been developed for the geometric calibration of multiple views containing projections of singularities of unknown positions. These methods have been adapted to the CB CT, e.g., see~\cite{konik21} with numerical experiments. However, BA needs iterative methods. Moreover, in tomography, the integral nature of the projection smooths out singularities, see~\cite{Natterer:86}, Chapter~II.5, p.42, and makes the singularity detection difficult. This is the reason why introducing opaque calibration markers in the scene simplifies the detection step. If the markers are in the field-of-view and their projections can be extracted in the X-ray projection images, then the geometric calibration can be done even if the X-ray projection images are truncated. This is a huge advantage.

Analytical solutions along with the design of particular calibration marker sets were recently proposed in~\cite{tischenko19} for the cone-beam geometry with a circular source trajectory and in~\cite{jonas18} for the fan-beam and cone-beam geometries with general source trajectories. In~\cite{tischenko19} the authors used a non-standard description of the geometry and required for their method two projections obtained after an accurate rotation of the calibration object by $180$ degrees. In~\cite{jonas18} non-standard descriptions of geometries were also considered for which parameters can be uniquely identified with linear systems of equations. In order to convert their geometrical parameters to standard ones, nonlinear inversions must be performed which are one-to-one in the ideal situation, but can cause problems in realistic settings. 

In digital subtraction angiography~\cite{unberath17}, DCCs are applied to the high contrast vascular tree projections (without introduction of markers). DCCs can be applied in the context of data truncation if singular objects of reduced support can be extracted in all projections (without being truncated). A similar idea has been already proposed in the context of motion estimation in projections and is based on mathematical properties of the projection of singularities~\cite{Katsevich2011} (it can be translated to a geometric calibration, because acquisition geometry modifications can be seen as object motions: e.g., a rigid transform of the acquisition system is equivalent to an inverse rigid movement of the measured object, see also~\cite{YuHengyong2007DCBR,LengShuai2007Mari}). However, these approaches need iterative methods for the parameter estimations.

\redC{In this work, we explore the application of data consistency conditions on distributions, applied to singularities resulting from opaque markers. Specifically, we address cases where only partial information is available from calibration marker sets.} \redC{We propose a simplified version of known DCCs on distributions in 2D parallel geometry, as well as new DCCs on distributions in 2D fan-beam geometry with sources on a line. The advantage of applying DCCs to distributions is that even when data are truncated, the DCCs of singularity projections can still be computed if these singularities remain within the field-of-view of each projection.}  We propose efficient closed-form formulas for the identification of large numbers of geometric calibration parameters.
\redC{The data are 2D images,} called sinograms in 2D parallel geometry and linograms in fan-beam geometry with X-ray sources on a line (see~\cite{desbat19,edholm87}). One axis is the projection angle in a sinogram or the source position parameter in a linogram; the second axis is the projection index for both geometries (see~\eref{eqn:radon2Ddef},~\eref{eqn:fanbeamdef} and \Fref{fig:geometries}).

Our paper is organized as follows. In the subsection~\ref{SecDistributions}, we  introduce our notation and recall basic facts on distributions. In the subsection~\ref{Sec2DRadon}, we recall the definitions of the 2D Radon transform of functions and distributions, known range conditions for this transform.
We also recall the fan-beam transform of functions and known range conditions for the fan-beam geometry when X-ray sources are on a line.
In the section~\ref{SecSelfCalib2DRadon}, we present \redC{a simplified version of DCCs on distributions and a possible application to} a geometric calibration problem from truncated data in 2D parallel geometry. We provide a special calibration marker set from which the detector shift and the angle of each projection are identified with an analytical formula.
Numerical simulations are presented at the end of the section.
The section~\ref{SecSelfCalibFanBeam} is dedicated to the fan-beam geometry with sources on a line.
In this geometry, we introduce \redC{new} necessary range conditions for multiple projections of distributions. \redC{One possible application of these new DCCs is the calibration with} the identification of the source position and the detector shift of each projection. With these proposed necessary range conditions for Dirac distributions in this geometry, we obtain a closed-form solution to the geometric calibration problem with a specific calibration marker set (with only partially known positions).
Numerical simulations are presented.
In the section~\ref{SecConclusions}, we discuss \redC{advantages and other perspectives of the application of DCCs on distributions}.

\section{Notation and basic facts}

\subsection{Distributions}
\label{SecDistributions}
Let $\Omega_N \subset \mathbb{R}^N$ denote an arbitrary open set, $N\in\mathbb{N}^\star$.
As in \cite{bony01,helgason99}, we will use:
\begin{itemize}
\item $\mathscr{D}\left(\Omega_N\right)$ the space  of smooth functions compactly supported in $\Omega_N$, $\mathscr{D}_N=\mathscr{D}\left( \mathbb{R}^N\right)$; 
\item $\mathscr{S}_N=\mathscr{S}\left( \mathbb{R}^N\right)$ the Schwartz space of rapidly decreasing smooth functions;  
\item  $\mathscr{E}\left( \Omega_N\right) = C^{\infty}\left(\Omega_N\right)$ the space of smooth functions on $\Omega_N$, $\mathscr{E}_N=C^{\infty}\left(\mathbb{R}^N\right)$;
\item $\mathscr{D}'_N$, $\mathscr{S}'_N$, $\mathscr{E}'_N$, $\mathscr{D}'\left( \Omega_N\right)$ and $\mathscr{E}'\left( \Omega_N\right)$ the spaces of all distributions acting on the test functions from $\mathscr{D}_N$, $\mathscr{S}_N$, $\mathscr{E}_N$, $\mathscr{D}\left( \Omega_N\right)$ and $\mathscr{E}\left( \Omega_N\right)$ respectively.
\end{itemize}

According to~\cite{helgason99}, we have the inclusions: $\mathscr{E}'_N \subset \mathscr{S}'_N \subset \mathscr{D}'_N$. Moreover, each function $f \in \mathscr{D}_N$ can be seen as the distribution $T_f$ acting on $\mathscr{E}_N$ with $T_f(\phi)=(f, \phi), \forall\phi\in\mathscr{E}_N$, where $(\cdot,\cdot)$ is the usual scalar product in $L^2\left(\mathbb{R}^N\right)$ defined by the integral
\begin{eqnarray}
    (f, \phi) \eqdef \int_{\mathbb{R}^N} f(\vec{x}) \phi(\vec{x}) \rmd\vec{x}.
\end{eqnarray}
Thus, $\mathscr{D}_N \subset \mathscr{E}'_N$. Also, each function $f \in \mathscr{E}_N$ can be seen as the distribution $T_f$ acting on $\mathscr{D}_N$ with $T_f(\phi)=(f, \phi), \forall\phi\in\mathscr{D}_N$, thus, $\mathscr{E}_N \subset \mathscr{D}'_N$.

We will focus on distributions from $\mathscr{E}'(\Omega_N)$. 
In the following sections~\ref{SecSelfCalib2DRadon} and~\ref{SecSelfCalibFanBeam}, we consider finite sums of Dirac distributions $\delta_{\vec{c}}$, $\vec{c}\in \Omega_N$,  acting as $\delta_{\vec{c}}(\phi) \eqdef \phi(\vec{c}), \forall\phi\in\mathscr{E}(\Omega_N)$. The distribution $\delta_{\vec{c}}$ is  of compact support $\left\{\vec{c}\right\}$,  $\delta_{\vec{c}}\in\mathscr{E}'(\Omega_N)$, see~\cite{gasquet99}.

According to \cite{bony01}, Chapter 7.2, p.135, a convolution of a distribution $T \in \mathscr{E}'_N$ with a smooth function of compact support $g$ can be defined as a function $f(\vec{x}) = \left(T(\vec{t}), g(\vec{x}-\vec{t})\right)$. Note that the action of the functional $T$ on the test function $\phi$ can be written as $T(\phi)$, $(T,\phi)$ or $\left( T(\vec{x}),\phi(\vec{x}) \right)$ if we need to emphasize that $T$ is a distribution acting on $\phi$ as a function of the variable $\vec{x}$. Moreover, we will use the following convolution properties:

\begin{theorem}[\cite{bony01}, 7.2, p.135]
Let $T \in \mathscr{E}'_N$, $g \in \mathscr{D}_N$, then $\supp(T *g) \subset \supp(T)+\supp(g)$.
\label{thm:21}
\end{theorem}

\begin{theorem}[\cite{bony01}, 7.2, p.135]
Let $T \in \mathscr{E}'_N$, $g \in \mathscr{D}_N$, then $T * g$ belongs to $\mathscr{D}_N$.
\label{thm:22}
\end{theorem}

Let us also recall that a sequence of distributions $\{T_j\}_{j \in \mathbb{N}}$ from $\mathscr{D}'_N$ converges to $T$ if for each $\phi \in \mathscr{D}_N$  $\lim_{j \to +\infty}(T_j, \phi)=(T, \phi)$, see~\cite{bony01}, Chapter 4.2, p.88. We also use the translation $\tau_{\vec{t}}$ by $\vec{t}\in\mathbb{R}^N$ of a distribution $T \in \mathscr{D}'_N$ defined  by $(\tau_{\vec{t}}T,\phi)\eqdef\left(T(\vec{x}-\vec{t}),\phi(\vec{x})\right)\eqdef\left(T(\vec{x}),\phi(\vec{x}+\vec{t})\right)=(T,\tau_{-\vec{t}}~\phi)$.

In the rest of the work $N=1$ or $N=2$.

\subsection{2D Radon and fan-beam transforms}
\label{Sec2DRadon} 
\begin{figure}
     \centering
     \begin{subfigure}[b]{0.49\textwidth}
         \centering
         \includegraphics[width=\textwidth]{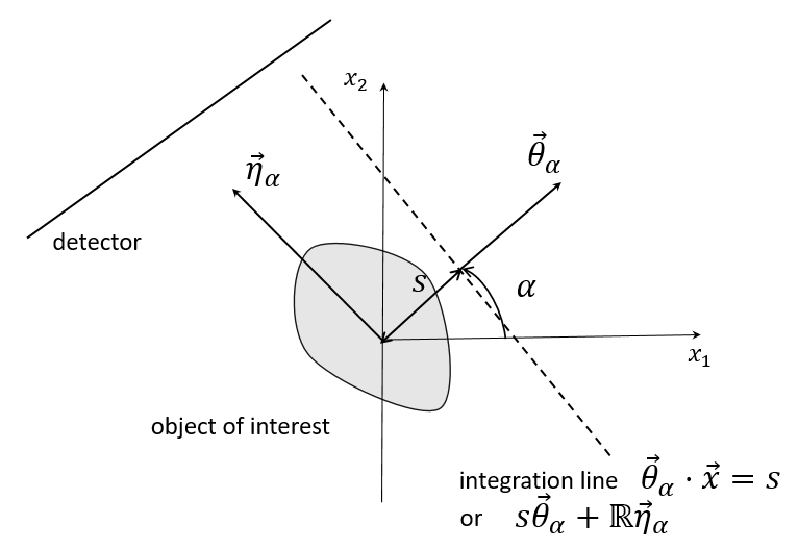}
         \caption{}
     \end{subfigure}
     \hfill
     \begin{subfigure}[b]{0.49\textwidth}
         \centering
         \includegraphics[width=\textwidth]{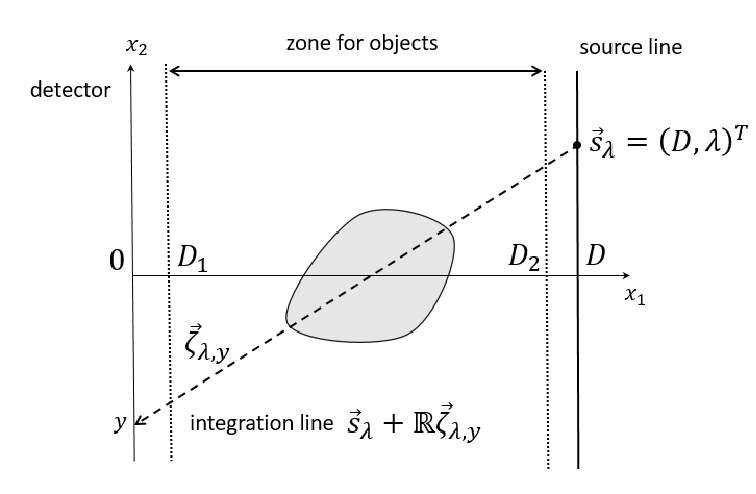}
         \caption{}
     \end{subfigure}
        \caption{(a) 2D parallel geometry, (b) fan-beam geometry with sources on a line.}
        \label{fig:geometries}
\end{figure}

We start with the 2D parallel geometry associated to the 2D Radon transform well known in tomography, see~\Fref{fig:geometries}~(a):
\begin{definition}
Let $f \in \mathscr{S}_2$, the Radon transform of $f$ is
\begin{eqnarray} 
\mathcal{R}f(\vec{\theta}_{\alpha},s)\eqdef \mathcal{R}f(\alpha,s) 
\eqdef \int_{\vec{\theta}_{\alpha} \cdot \vec{x}=s} f(\vec{x})\rmd \vec{x} = \int_{-\infty}^{+\infty} f(s\vec{\theta}_{\alpha}+l\vec{\eta}_{\alpha})\rmd l,
\label{eqn:radon2Ddef}
\end{eqnarray}
where  $s \in \mathbb{R}$, $\vec{\theta}_{\alpha}\eqdef (\cos{\alpha},\sin{\alpha})^T,$ $\vec{\eta}_{\alpha}\eqdef(-\sin{\alpha}, \cos{\alpha})^T$, $\alpha \in [0, 2\pi)$ ($\vec{\theta}_{\alpha}$ and $\vec{\eta}_{\alpha}$ belong to the unit circle $S^{1}$), $\vec{\theta}_{\alpha} \cdot \vec{x}=x_1 \cos{\alpha} + x_2 \sin{\alpha}$ is the usual inner product of two vectors.
\end{definition}

\begin{theorem}[Range conditions for $\mathcal{R}$ on functions,~\cite{Natterer:86}, II.4, p.36]
A function $p$ is the Radon transform of $f \in \mathscr{S}_2$ if and only if:
\begin{enumerate}
    \item $p \in \mathscr{S}\left(S^{1} \times \mathbb{R}\right)$,
    \item $p$ is even: $\forall s \in \mathbb{R}$, $\forall \vec{\theta}_{\alpha} \in S^{1}$ $p(-\vec{\theta}_{\alpha},-s)=p(\vec{\theta}_{\alpha},s)$,
    \item for $k=0,1,2,\dots$, for all $\vec{\theta}_{\alpha} \in S^1$ we have the moment conditions:
    \begin{eqnarray}
        \int_{-\infty}^{+\infty} p(\vec{\theta}_{\alpha},s) s^k \rmd s=\mathscr{P}_k(\vec{\theta}_{\alpha}),
    \label{eqn:hlccmomentcond}
    \end{eqnarray}
    $\mathscr{P}_k(\vec{\theta}_{\alpha})$ is a homogeneous polynomial of degree at most $k$ in the coordinates $\cos{\alpha}$ and $\sin{\alpha}$ of $\vec{\theta}_{\alpha}$.
\end{enumerate}
Moreover, $p(\vec{\theta}_{\alpha},s)=0$ for $|s| > a \Leftrightarrow f(\vec{x})=0$ for $\norm{\vec{x}} > a$.
\label{thm:hlcc}
\end{theorem}

This theorem expresses data consistency conditions for the Radon transform, thus, provides the description of the range of the Radon transform when $f$ and $p$ are functions. These range conditions are known in the literature as Helgason-Ludwig consistency conditions (HLCCs)~\cite{helgason65,ludwig66}, see also~\cite{gelfand66}.

The generalization of the Radon transform of functions to distributions in arbitrary dimension was defined by Gelfand and colleagues in~\cite{gelfand66} in the middle of the last century. In~\cite{ramm96} the authors introduced DCCs for the Radon transform on distributions. The Radon transform of $f$ from $\mathscr{D}'_2$, $\mathscr{S}'_2$, $\mathscr{E}'_2$ is defined with the duality relation
\begin{eqnarray}
    (\mathcal{R}f,\phi)=\langle f,\mathcal{R}^{*} \phi \rangle, \: \mathcal{R}^{*} \phi(\vec{x})=  \int_{S^1} \phi(\vec{\theta}_{\alpha}, \vec{x} \cdot \vec{\theta}_{\alpha}) \rmd \vec{\theta}_{\alpha},
    \label{eqn:duality}
\end{eqnarray}
where $\phi$ is a test function on $S^1 \times \mathbb{R}$  of the corresponding space. If $f \in \mathscr{S}_2$, then $(\cdot,\cdot)$ represents the scalar product in $L^2(S^1 \times \mathbb{R})$, $\langle \cdot,\cdot \rangle$ represents the scalar product in $L^2(\mathbb{R}^2)$. For distributions of compact support the full analogue of HLCCs was proven by Ramm and Katsevich, see~\cite{ramm96}. We recall the result in 2D: 
\begin{theorem}[Range conditions for $\mathcal{R}$ on distributions,~\cite{ramm96}, 10.4, p.313]
\label{TheoRammAndKat}
A distribution $p$ is the Radon transform of $f \in \mathscr{E}'_2$ if and only if:
\begin{enumerate}
    \item $p \in  \mathscr{E}'\left(S^1 \times \mathbb{R}\right)$,
    \item $p$ is even: distributions $p(-\vec{\theta}_{\alpha},-s)$ and $p(\vec{\theta}_{\alpha},s)$ act identically,
    \item for $k=0,1,2,\dots$, $\forall \psi \in C^{\infty} (S^1)$ we have the moment conditions:
    \begin{eqnarray}
        \left(p(\vec{\theta}_{\alpha},s), s^k \psi(\vec{\theta}_{\alpha})\right)=\int_{S^{1}} \mathscr{P}_k(\vec{\theta}_{\alpha})\psi(\vec{\theta}_{\alpha}) \rmd \vec{\theta}_{\alpha},
    \label{eqn:dccRKmomentcond}
    \end{eqnarray}
    $\mathscr{P}_k(\vec{\theta}_{\alpha})$ is a homogeneous polynomial of degree at most $k$ in the coordinates $\cos{\alpha}$ and $\sin{\alpha}$ of $\vec{\theta}_{\alpha}$.
\end{enumerate}
Moreover, if $p(\vec{\theta}_{\alpha},s)=0$ for $|s| > a \Leftrightarrow \supp (f)$ is in the ball of the radius $a$ centered at the origin.
\label{thm:dccRK}
\end{theorem}

The Radon transform, i.e., the parallel geometry, is rather a theoretical tool. The divergent beam transform such as the 2D fan-beam transform is more common in applications:
\begin{definition} The fan-beam transform of the function $f \in \mathscr{S}_2$ is defined by 
\begin{eqnarray} \mathcal{D}f(\lambda,\alpha) \eqdef \int_0^{+ \infty} f(\vec{s}_{\lambda}+l\vec{\zeta}_{\alpha})\rmd l, 
\end{eqnarray}
where  $\lambda\in\mathbb{R}$ is the trajectory parameter of the source $\vec{s}_{\lambda}\in \mathbb{R}^2$, $\vec{\zeta}_{\alpha}\in \mathbb{R}^2$ is the direction vector of the integration beam line $\vec{s}_{\lambda}+\mathbb{R}\vec{\zeta}_{\alpha}$.
\end{definition}
It's usually required that $\vec{\zeta}_{\alpha}\in S^1$, i.e., is a unit vector. It's also classical for the tomosynthesis geometry to consider $\vec{\zeta}_{\alpha}$ to be a non-unit vector (thus, to consider a weighted fan-beam transform), see~\cite{clackdoyle13}. 
When the source is on the line $x_1=D$, $D>0$, see \Fref{fig:geometries}~(b), the source trajectory is given by $\vec{s}_{\lambda}=(D,\lambda)^T$, $\lambda\in\mathbb{R}$. As in~\cite{clackdoyle13},  we consider that the detector is the $x_2$-axis. The direction vector of the integration line is parameterized by $y\in\mathbb{R}$, the position on the $x_2$-axis of a detector unit. Thus, the direction of the integration line $\vec{\zeta}_{\lambda,y}$ is given by $(0,y)^T-(D,\lambda)^T$. The integration line is $\vec{s}_{\lambda}+\mathbb{R}\vec{\zeta}_{\lambda,y}=(D, \lambda)^T+\mathbb{R}(-D,y-\lambda)^T$. In this case, the fan-beam data of the function $f$ is
\begin{eqnarray}
    \mathcal{D}f(\lambda,y) \eqdef \int_0^{+ \infty} f(D-lD, \lambda+ly-l\lambda)\rmd l.
    \label{eqn:fanbeamdef}
\end{eqnarray}
The following moment conditions were shown in \cite{clackdoyle13} for this geometry:
\begin{theorem}
Define $\mathscr{P}_k(\lambda)=\int_{- \infty}^{+ \infty} g(\lambda,y)y^k \rmd y$ for all $k=0,1,2,...$ The function $\mathscr{P}_k(\lambda)$ is a polynomial in $\lambda$ of degree $k$ if and only if $g=\mathcal{D}f$ for some function $f$.
\label{thm:dccclackdoyle}
\end{theorem}
We consider here $f$ as a smooth function of compact support. The support of $f$ $\supp(f)$ is completely in between the detector line and the source line in practice. To the best of our knowledge, DCCs for the fan-beam transform on distributions weren't discussed in the literature.

\section{\redC{Range conditions for the 2D Radon transform on distributions and geometric calibration}
}
\label{SecSelfCalib2DRadon}
\subsection{Mathematical results}

We first want to derive a simplified version of the known moment conditions~\eref{eqn:dccRKmomentcond}  of Theorem~\ref{TheoRammAndKat} for distributions. We consider the 2D Radon transform at fixed $\alpha$ and derive DCCs similar to those of Theorem~\ref{TheoRammAndKat}.
As in~\cite{Natterer:86}, we consider 
$\mathcal{R}_{\alpha}f(s)\eqdef\mathcal{R}f(\alpha,s)$, $\alpha\in[0,2\pi)$, $s\in\mathbb{R}$ for $f \in \mathscr{S}_2$, thus, $\mathcal{R}_{\alpha}f$ is a function of one variable $s$ for $\alpha$ fixed. 
For $f \in \mathscr{S}_2$ and $\phi \in \mathscr{S}_1$ with the classical change of variables $\vec{x}=s\vec{\theta}_{\alpha}+l\vec{\eta}_{\alpha}$:
\begin{eqnarray}
    \fl (\mathcal{R}_{\alpha}f,\phi)=\int_{-\infty}^{+\infty} \mathcal{R}_{\alpha}f(s) \phi(s) \rmd s = \int_{-\infty}^{+\infty} \int_{-\infty}^{+\infty} f(s\vec{\theta}_{\alpha}+l\vec{\eta}_{\alpha})\rmd l \phi(s) \rmd s \nonumber\\
    = \int_{\mathbb{R}^2} f(\vec{x}) \phi(\vec{x} \cdot \vec{\theta}_{\alpha}) \rmd \vec{x}= \langle f,\mathcal{R}_{\alpha}^{*} \phi \rangle,
    \label{eqn:dualityradonfuncour}
\end{eqnarray}
where $(\cdot,\cdot)$ is the canonical duality pairing for test functions on $\mathbb{R}$, $\langle \cdot,\cdot \rangle$ is the canonical duality pairing for test functions on $\mathbb{R}^2$. Thus, $\forall\alpha\in[0,2\pi)$ the adjoint operator $\mathcal{R}_{\alpha}^{*}$ for functions in $\mathscr{S}_1$:
\begin{eqnarray}
   \mathcal{R}_{\alpha}^{*} \phi(\vec{x}) = \phi(\vec{x} \cdot \vec{\theta}_{\alpha}), \: \forall\vec{x}\in\mathbb{R}^2.
   \label{eqn:adjointradon}
\end{eqnarray}
The definition \eref{eqn:adjointradon} can be generalized to $\phi \in \mathscr{E}_1$. We have $\mathcal{R}_{\alpha}^{*} \phi \in \mathscr{E}_2$ as a composition of two smooth functions $\vec{x}\mapsto \vec{x} \cdot \vec{\theta}_{\alpha}$ and $\phi$. The parallel projection of a distribution is classically defined by:
\begin{definition}
 The Radon transform $\mathcal{R}_{\alpha} f$ at fixed $\alpha\in [0,2\pi)$ of $f \in \mathscr{E}'_2$ is a distribution acting on the space $\mathscr{E}_1$ of test functions according to 
 \begin{eqnarray}
     (\mathcal{R}_{\alpha}f,\phi) \eqdef \langle f,\mathcal{R}_{\alpha}^{*} \phi \rangle.
     \label{eqn:dualityradonour}
 \end{eqnarray}
\end{definition}

The linearity of the operator $\mathcal{R}_{\alpha} f$ is obvious. The continuity is equivalent with its boundedness. This fact is explained and a proof of the boundedness can be found in Appendix~\ref{SecAppendix}, see Lemma~\ref{LemmaRadonBounded}. It's the proof of the first point~{\em(\ref{item_one_of_ParallelRangeCondTh})} in the next theorem.
\begin{theorem}[Necessary range conditions for $\mathcal{R}_{\alpha}$ on distributions, $\mathbf{\alpha\in[0,2\pi)}$]
If $f \in \mathscr{E}'_2$, the distribution $p_{\alpha}$ is the Radon transform $\mathcal{R}_{\alpha}f$ of the distribution $f$ at fixed $\alpha\in[0,2\pi)$, then:
\begin{enumerate}
    \item \label{item_one_of_ParallelRangeCondTh} $p_{\alpha} \in  \mathscr{E}'_1$,
    \item \label{item_two_of_ParallelRangeCondTh} $p_{\alpha}$ is even: the distributions  $p_{\alpha+\pi}(-s)$ and $p_{\alpha}(s)$ act identically on $\mathscr{E}'_1$ (or the corresponding distributions $p_{-\vec{\theta}_{\alpha}}(-s)$ and $p_{\vec{\theta}_{\alpha}}(s)$ act identically),
    \item \label{item_three_of_ParallelRangeCondTh} for $k=0,1,2,\dots$, $\forall \alpha$ the moment conditions are necessary fulfilled:
    \begin{eqnarray}
        \left(p_{\alpha} (s), s^k\right)=\mathscr{P}_k(\alpha),
    \label{eqn:dccradonourmomentcond}
    \end{eqnarray}
    where $\mathscr{P}_k(\alpha)$ is a homogeneous polynomial of degree at most $k$ in $\cos{\alpha}$, $\sin{\alpha}$.
\end{enumerate}
Moreover, $p_{\alpha}(s)=0$ for $|s| > a$ if $\supp (f)$ is in the ball $B_a$ of the radius $a$ centered at the origin.
\label{thm:dccradonour}
\end{theorem}

\begin{proof}
The first point~{\em(\ref{item_one_of_ParallelRangeCondTh})} is proven in Appendix~\ref{SecAppendix}, Lemma~\ref{LemmaRadonBounded}.

For~{\em(\ref{item_two_of_ParallelRangeCondTh})}, for any $f \in \mathscr{E}'_2$ with $\supp (f) \subset B_a$, we classically define the function $f_{\epsilon}=W_{\epsilon} *f$, where $W_{\epsilon}\eqdef\epsilon^{-2}W_1(\vec{x}/\epsilon)$, $W_1 \in \mathscr{D} (B_1)$ is positive and $\int_{B_1} W_1 (\vec{x}) \rmd \vec{x}=1$. The function $f_{\epsilon} \in \mathscr{D} (B_{a+\epsilon})$ from Theorem \ref{thm:21} and  Theorem  \ref{thm:22}. From HLCCs (Theorem \ref{thm:hlcc}), 
$\mathcal{R}f_{\epsilon} \in  \mathscr{D} (Z_{2,a+\epsilon})$, where $Z_{2,a+\epsilon}=[0, 2\pi] \times [-a-\epsilon, a+\epsilon]$. At fixed $\alpha$, $\mathcal{R}_{\alpha}f_{\epsilon} \in \mathscr{D} ([-a-\epsilon,a+\epsilon])$. The same integral equality~\eref{eqn:dualityradonfuncour} is applied to the function $\mathcal{R}_{\alpha}f_{\epsilon}$ and $\forall \phi \in \mathscr{E}_1$
\begin{eqnarray}
    (\mathcal{R}_{\alpha}f_{\epsilon},\phi)=\langle f_{\epsilon},\mathcal{R}_{\alpha}^{*} \phi \rangle.
    \label{eqn:cond1}
\end{eqnarray}
According to~\cite{bony01}, Chapter 7.2, p.137, for such regularization function $W_{\epsilon}$ we have
\begin{eqnarray}
    \lim_{\epsilon \to 0} \langle f_{\epsilon},\mathcal{R}_{\alpha}^{*} \phi \rangle =\langle f,\mathcal{R}_{\alpha}^{*} \phi \rangle.
    \label{eqn:cond2}
\end{eqnarray}
From~\eref{eqn:cond1},~\eref{eqn:cond2} and the definition $(\mathcal{R}_{\alpha}f,\phi)=\langle f,\mathcal{R}_{\alpha}^{*} \phi \rangle$ \ $\lim_{\epsilon \to 0} \mathcal{R}_{\alpha}f_{\epsilon} = \mathcal{R}_{\alpha}f$. Thus, the properties of $\mathcal{R}_{\alpha}f_{\epsilon}$ can be transferred to $\mathcal{R}_{\alpha}f$. In particular,  $\mathcal{R}_{\alpha}f(s)=0$ for $|s| > a$, $\mathcal{R}_{\alpha}f$ is even (in the sense as described before), so~{\em(\ref{item_two_of_ParallelRangeCondTh})} is proven.  

For the moment conditions~{\em(\ref{item_three_of_ParallelRangeCondTh})}, with  $s^k \in \mathscr{E}_1$, $k\in\mathbb{N}$:
\begin{eqnarray}
 \fl \left(\mathcal{R}_{\alpha}f(s), s^k\right)=\Braket{ f(\vec{x}),\mathcal{R}_{\alpha}^{*} (s^k)(\vec{x})}=\Braket{ f(\vec{x}), (\vec{x} \cdot \vec{\theta}_{\alpha})^k  }  = \Braket{ f(\vec{x}), (x_1 \cos{\alpha}+x_2 \sin{\alpha})^k } \nonumber\\ = \Braket{ f(\vec{x}), \sum_{i=0}^{k}\binom{k}{i} (x_1   \cos{\alpha})^{k-i} (x_2 \sin{\alpha})^{i} }  \nonumber\\ = \sum_{i=0}^{k} \binom{k}{i} \Braket{ f(\vec{x}), x_1^{k-i} x_2^{i}} (\cos{\alpha})^{k-i} (\sin{\alpha})^{i} = \mathscr{P}_k(\alpha).
\end{eqnarray}
\end{proof}

In the following, a small, highly attenuating marker at $\vec{c}\in \mathbb{R}^2$ will be modeled by a Dirac distribution $\delta_{\vec{c}}$ of compact support. The  Radon transform at fixed $\alpha$ of $\delta_{\vec{c}}$ is given by
 \begin{eqnarray}
  \fl \left(\mathcal{R}_{\alpha}\delta_{\vec{c}}(s), \phi(s)\right)=\Braket{ \delta_{\vec{c}}(\vec{x}) ,\mathcal{R}_{\alpha}^{*} \phi(\vec{x}) }=\Braket{ \delta_{\vec{c}} (\vec{x}), \phi(\vec{x} \cdot \vec{\theta}_{\alpha})  }  =  \phi(\vec{c} \cdot \vec{\theta}_{\alpha}) \nonumber \\ = \delta_{\vec{c} \cdot \vec{\theta}_{\alpha}} (\phi).
\end{eqnarray}
The Radon transform at fixed $\alpha$ of the Dirac distribution at $\vec{c}$ is the Dirac distribution at $\vec{c} \cdot \vec{\theta}_{\alpha}$ in the projection.  \redC{Then~\eref{eqn:dccradonourmomentcond} gives us the idea that we can compute moments of the projection data in this case and they will have the same homogeneous polynomial behavior in $\cos \alpha$ and $\sin\alpha$ as for functions in~\eref{eqn:hlccmomentcond}}.

In a realistic setting, a marker isn't a Dirac distribution, it's a small ball, a real physical object. The Radon transform for the characteristic function $\chi_{B_{(\vec{c},R)}}$ of the  circular disc $B_{(\vec{c},R)}$ of radius $R>0$ centered at $\vec{c}=(c_1,c_2)^T$ (using the shifting property, e.g., see~\cite{deans10})  is
\begin{eqnarray}
    \mathcal{R}\chi_{B_{(\vec{c},R)}}(\alpha, s)= \cases{ 2\sqrt{R^2-(s-s_0(\alpha,c_1,c_2))^2} & if $|s-s_0|\leq R$ \\ 0 & otherwise,}
    \label{eqn:radonofdisk}
\end{eqnarray}
where $s_0(\alpha,c_1,c_2)=c_1\cos{\alpha}+c_2\sin{\alpha}=\vec{c}\cdot\vec{\theta}_{\alpha}$. The support of $\mathcal{R}_{\alpha}\chi_{B_{(\vec{c},R)}}$ is $[s_0-R, s_0+R]$, so our projected Dirac distribution $\delta_{\vec{c} \cdot \vec{\theta}_{\alpha}}$ is localized at the center $s_0$ of the $\alpha$-projection of the round marker. 

In the next subsection, we assume that $n$ markers are in the field-of-view of the scanner (in general, they have been introduced). Then the projected object is modeled by $\fob+\fballs$, where $\fob$ is the object attenuation function and $\fballs$ is the attenuation function of the markers (in general, small stainless steel balls). In the calibration process, we then replace $\fballs$ by $f=\sum_{j=1}^{n} \delta_{\vec{c}_j}$, where $\vec{c}_j \in\mathbb{R}^2$, $j=1,\ldots,n$, are the centers of the markers, and from the linearity of the Radon transform $\mathcal{R}_{\alpha}f=\sum_{j=1}^{n} \delta_{\vec{c}_j \cdot \vec{\theta}_{\alpha}}$. 

\subsection{Geometric calibration problem to solve}

Existing calibration methods for the 2D Radon transform derived from HLCCs are based on the moments of the projections~\eref{eqn:hlccmomentcond}. We need to compute the integrals from $-\infty$ to $+\infty$ over $s$ of the measured projection data. Thus, we can't work with truncated projections. For example, in~\Fref{fig:truncationradon}, the object of interest (the ellipse) isn't completely in the field-of-view. Projection data are truncated. We don't have all non-zero projection data. Moments in~\eref{eqn:hlccmomentcond} can not be computed.

\begin{figure}
     \centering
     \begin{subfigure}[b]{0.48\textwidth}
         \centering
         \includegraphics[width=\textwidth]{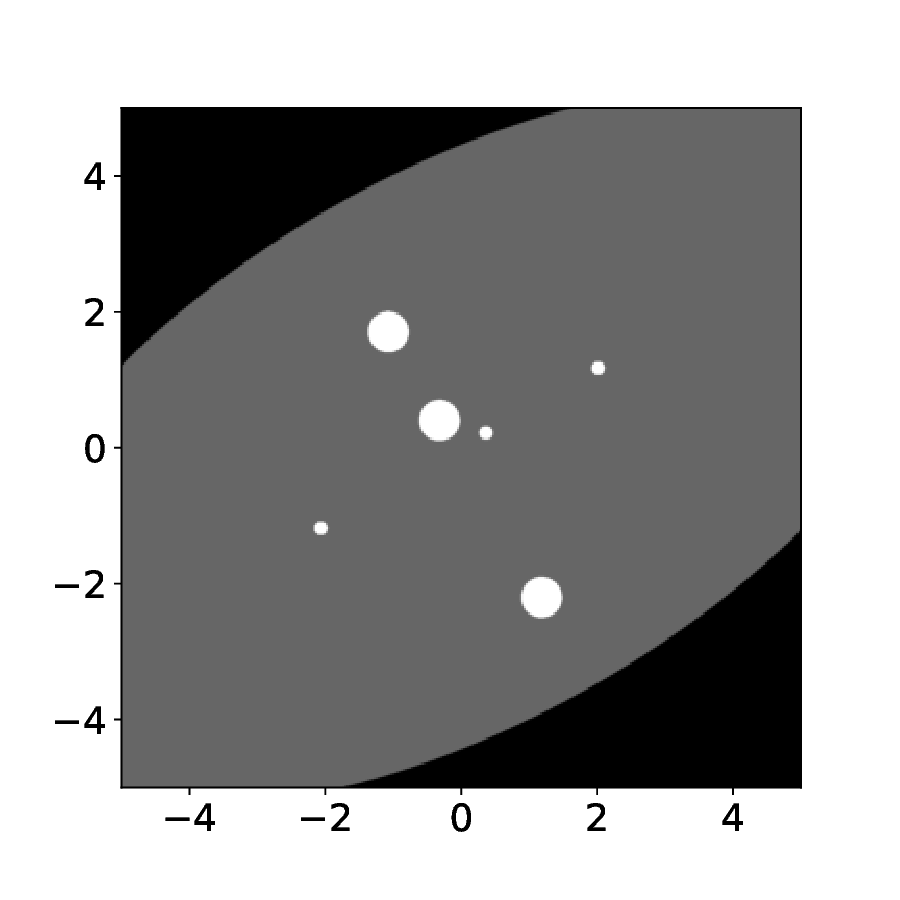}
         \caption{}
     \end{subfigure}
     \hfill
     \begin{subfigure}[b]{0.48\textwidth}
         \centering
         \includegraphics[width=\textwidth]{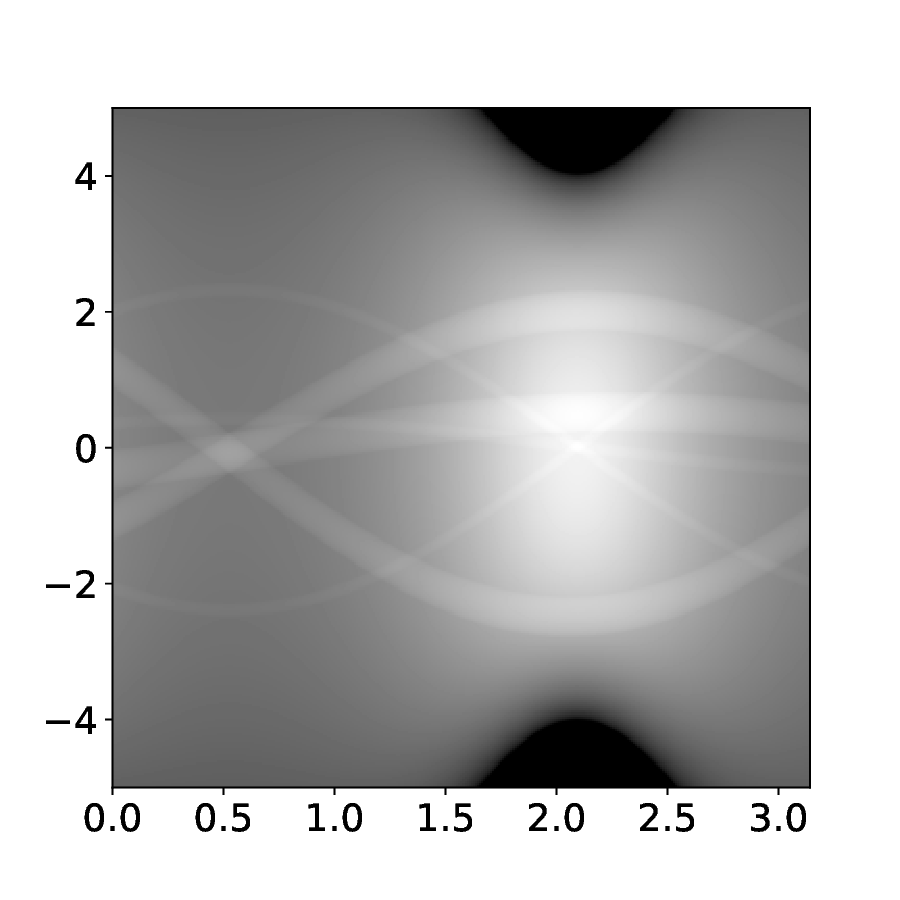}
         \caption{}
     \end{subfigure}
        \caption{(a) measured data $\fob+\fballs$: truncated object $\fob$ and round markers $\fballs$ in $\mathbb{R}^2$ (the horizontal axis is $x_1$, the vertical axis is $x_2$), (b) their {\em truncated} parallel projections (the horizontal axis is $\alpha$ (here $\alpha$ is measured according to the $x_1$ horizontal axis of the image (a)), the vertical axis is $s$). At fixed $\alpha$ each projection $\mathcal{R}_{\alpha}(\fob+\fballs)$ is truncated, but $\mathcal{R}_{\alpha}(\fballs)$ is {\em not truncated}.}
        \label{fig:truncationradon}
\end{figure}

We address here the shift and angle calibration problem, see~\cite{basu100}. We assume that  projection angles $\alpha_0, ... , \alpha_{P-1}$ of $P\in\mathbb{N}^\star$ projections are unknown and that each projection has been shifted by an unknown shift $s_{\alpha_i}\in\mathbb{R}$,
i.e.,  we have the measured data $\overline{m}_{i}(s)=\mathcal{R}_{\alpha_i}(\fob+\fballs)(s-s_{\alpha_i})$, $i=0,\ldots,P-1$, $s\in\mathbb{R}$, with both $\alpha_i$ and $s_{\alpha_i}$ unknown ($i=0,\ldots,P-1$, thus, $2P$ unknowns).
We assume that we can detect with some precision the (shifted) center projections $v_{ij}\in\mathbb{R}$ of $n$ markers in each projection:
\begin{eqnarray}
v_{ij}=\vec{c}_j \cdot \vec{\theta}_{\alpha_i}+s_{\alpha_i}, \: j=1,\ldots,n, \:  i=0,\ldots,P-1.
\label{eqn:vijdef}
\end{eqnarray} 
Thus, we can compute and introduce the measurement model in the distribution form 
\begin{eqnarray}
m_{i}(s)=\sum_{j=1}^{n} \delta_{v_{ij}} (s), \: i=0,\ldots,P-1,
\label{eqn:misvij}
\end{eqnarray} 
being just the projection of the markers modeled by the sum of Dirac distributions $\sum_{j=1}^{n} \delta_{\vec{c}_j}$, 
thus,  
\begin{eqnarray}
 m_{i}(s)=\mathcal{R}_{\alpha_i}\left(\sum_{j=1}^{n} \delta_{\vec{c}_j}\right)\left(s-s_{\alpha_i}\right)=\sum_{j=1}^{n} \delta_{\vec{c}_j \cdot \vec{\theta}_{\alpha_i}}\left(s-s_{\alpha_i}\right),
\label{eqn:mis}
\end{eqnarray} 
where $\vec{c}_j$ is the (unknown) true 2D center of the marker $j$.

In the next subsection, \redC{inspired by the moment conditions~\eref{eqn:dccradonourmomentcond}, we will leverage the moments of the projection data of Dirac distributions $m_i$ to develop a calibration algorithm based only on the information extracted from the non-truncated marker projections (the real projections $\overline{m}_i$ may be truncated).}  To solve the angle and shift calibration problem, we propose here a specific calibration marker set of $n$ balls.  We don't know the marker positions $\vec{c}_j$, $j=1,\ldots, n$, but we assume that the markers are lying on two perpendicular lines (here $\nv$ markers are supposed to be on one line, other $\nh$ markers are on a perpendicular line, $n = \nv+\nh$). The calibration task is to compute all $\alpha_i$ and $s_{\alpha_i}$, $i=0,\ldots,P-1$, based on the detected projected markers $v_{ij}$ in the data ${m}_i$, $i=0,\ldots,P-1$, $ j=1,\ldots, n$.

\subsection{Calibration method}

Let us first remind that the calibration problem based only on the measured data doesn't have a unique solution for functions.   After a rotation $R_{\gamma}$ (of angle $\gamma\in \mathbb{R}$) and a translation $\vec{t}\in \mathbb{R}^2$ of the function  $f$ we still have the same data, but for other acquisition parameters:
\begin{eqnarray}
    \mathcal{R}f_{R_{\gamma},\vec{t}} \left(\alpha-\gamma,s-\vec{\theta}_{\alpha}\cdot \vec{t}\right)=\mathcal{R}f (\alpha,s),
\end{eqnarray}
where $f_{R_{\gamma},\vec{t}}(\vec{x}) \eqdef f\left(R_{\gamma}\vec{x}+\vec{t}\right)$, $R_{\gamma} \eqdef \left( \matrix{ \cos{\gamma} & -\sin{\gamma} \cr
\sin{\gamma} & \cos{\gamma}} \right)$. 
Moreover, if  $f_{\sigma_1}(x_1,x_2) \eqdef f(x_1,-x_2)$, then $\mathcal{R}f_{\sigma_1}(\alpha, s) = \mathcal{R}f(-\alpha, s)$. Combining this with a rotation yields the symmetry ambiguity according to any axis, see~\cite{basu100}.
Thus, from the projection data only we can not identify the function $f$ better than up to a rotation, a translation and a symmetry and the acquisition parameters accordingly.

We can associate to $f \in \mathscr{E}'_2$ the distribution $f_{R_{\gamma},\vec{t}}$ defined by  $\Braket{ f_{R_{\gamma},\vec{t}}(\vec{x}), \phi(\vec{x}) }=\Braket{ f(\vec{x}), \phi\left(R_{\gamma}^{-1}(\vec{x}-\vec{t})\right) }$.
It's easy to see that ${(\delta_{\vec{c}})}_{R_{\gamma},\vec{t}}$ is the distribution $\delta_{R_{\gamma}^{-1}(\vec{c}-\vec{t})}$:
\begin{eqnarray}
\fl\Braket{ {(\delta_{\vec{c}})}_{R_{\gamma},\vec{t}}(\vec{x}), \phi(\vec{x}) }
=  \Braket{ \delta_{\vec{c}} (\vec{x}), \phi\left(R_{\gamma}^{-1}(\vec{x}-\vec{t})\right) } = \phi\left(R_{\gamma}^{-1}(\vec{c}-\vec{t})\right)\nonumber\\ =\Braket{ \delta_{R_{\gamma}^{-1}(\vec{c}-\vec{t})} (\vec{x}), \phi(\vec{x}) }.
\label{eqn:distributionRT}
\end{eqnarray}
Then we have the same ambiguity for the Dirac distribution $f=\delta_{\vec{c}}$ (thus, also for finite sums of Dirac distributions):
\begin{eqnarray}
\fl \left(\mathcal{R}_{\alpha}{(\delta_{\vec{c}})}_{R_{\gamma},\vec{t}}(s), \phi(s)\right)= \Braket{ {(\delta_{\vec{c}})}_{R_{\gamma},\vec{t}}(\vec{x}), \phi(\vec{x} \cdot \vec{\theta}_{\alpha}) } = \Braket{ \delta_{R_{\gamma}^{-1}(\vec{c}-\vec{t})} (\vec{x}), \phi(\vec{x} \cdot \vec{\theta}_{\alpha}) } \nonumber\\ 
= \phi\left((R_{\gamma}^{-1}\vec{c}-R_{\gamma}^{-1}\vec{t})\cdot \vec{\theta}_{\alpha}\right)=\phi\left((\vec{c}-\vec{t})\cdot R_{\gamma}\vec{\theta}_{\alpha}\right)=\phi\left((\vec{c}-\vec{t})\cdot \vec{\theta}_{\alpha'}\right) \nonumber\\ =\left( \delta_{\vec{c} \cdot \vec{\theta}_{\alpha'}} (s), \phi(s-\vec{t}\cdot \vec{\theta}_{\alpha'})\right) =\left(\mathcal{R}_{\alpha'} \delta_{\vec{c}}(s), \phi(s-\vec{t}\cdot \vec{\theta}_{\alpha'})\right) \nonumber\\ =\left(\mathcal{R}_{\alpha'} \delta_{\vec{c}}(s+\vec{t}\cdot \vec{\theta}_{\alpha'}), \phi(s)\right),
\end{eqnarray}
where $\alpha'=\alpha+\gamma$. We have equivalently
\begin{eqnarray}
    \left(\mathcal{R}_{\alpha-\gamma}\left({(\delta_{\vec{c}})}_{R_{\gamma},\vec{t}}\right)\right)(s-\vec{\theta}_{\alpha}\cdot \vec{t})=\mathcal{R}_{\alpha}\delta_{\vec{c}} (s).
\end{eqnarray}
Thus, we have the same projection data as at angle $\alpha$ and $s$ after a rotation of angle $\gamma$ and a translation $\vec{t}$ of a Dirac, but with other parameters, namely at angle $\alpha-\gamma$ and at $s-\vec{\theta}_{\alpha}\cdot \vec{t}$. With projection data only, 
we have a non-unique solution of the geometric calibration problem. Since we have these degrees of freedom "up to a rigid transform" for functions and for Dirac distributions, we can choose to work in any coordinate system and identify the solution accordingly.  

As for functions, we consider the symmetry $(\delta_{\vec{c}})_{\sigma_1}$: ${(\delta_{\vec{c}})}_{\sigma_1}=\delta_{(c_1,-c_2)}$.
We have the ambiguity for the Radon transform:
\begin{eqnarray}
\fl (\mathcal{R}_{\alpha}{(\delta_{\vec{c}})}_{\sigma_1}(s), \phi(s))= \langle {(\delta_{\vec{c}})}_{\sigma_1}(\vec{x}), \phi(\vec{x} \cdot \vec{\theta}_{\alpha}) \rangle = \langle \delta_{(c_1,-c_2)} (\vec{x}), \phi(\vec{x} \cdot \vec{\theta}_{\alpha}) \rangle \nonumber\\ 
= \phi(c_1\cos{\alpha}-c_2\sin{\alpha})=\phi(\vec{c}\cdot \vec{\theta}_{-\alpha}) =\langle \delta_{\vec{c}} (\vec{x}), \phi(\vec{x}\cdot \vec{\theta}_{-\alpha}) \rangle\nonumber \\=(\mathcal{R}_{-\alpha} \delta_{\vec{c}}(s), \phi(s)).
\label{eqn:symmetry1}
\end{eqnarray}
Combining with the angular ambiguity, a symmetry according to any axis leads to similar results. In particular, for $\sigma_{2}$ such that $(\delta_{\vec{c}})_{\sigma_2}=\delta_{(-c_1,c_2)}$:
\begin{eqnarray}
    \mathcal{R}_{\alpha}{(\delta_{\vec{c}})}_{\sigma_2}=\mathcal{R}_{\pi-\alpha} \delta_{\vec{c}}.
\label{eqn:symmetry2}
\end{eqnarray}

In our calibration algorithm, we always work in coordinate systems with the same perpendicular directions, each is associated with a group of aligned markers. The first direction will be called "horizontal" and  is associated with the group of the so-called "horizontal" aligned markers. The perpendicular direction will be called "vertical" and is associated with the group of the so-called "vertical" aligned markers. In the following, we use the superscript $l=\rmh$ or $l=\rmv$, $\rmh$ will correspond to the "horizontal" group of markers and $\rmv$ to the "vertical" one. We will use several coordinate system centers in order to make the geometric calibration parameter identification more efficient. As we will see, it simplifies the angle identification with the moments of order 2 and 3 (for each group of markers independently).

We assume that we can detect in the projections $\nh$ centers of the "horizontal" markers $\vec{c}^{\:\rmh}_j$, $j=1,\ldots,\nh$, and  $\nv$ centers of the "vertical" markers $\vec{c}^{\:\rmv}_j$, $j=1,\ldots,\nv$, the total number of markers is $n$.  For example, the radii $R^{\:\rmh}$ of the "horizontal" marker disks can be chosen smaller compare to the radii $R^{\:\rmv}$ of the "vertical" markers as in~\Fref{fig:truncationradon}~(a).  
In practical experiments, the horizontal or vertical projections of the horizontal or vertical marker centers are identical. But with the size difference of the horizontal and vertical markers we can easily decide (up to $\pi$) if the projection is horizontal or vertical. More generally, directions close to these projection angles can induce detection problems as marker projections then overlap in directions close to horizontal or to vertical. But here again we know that the projection direction is close to either horizontal or vertical. Anyway, this is a limitation of the proposed approach due to the detection process, which we haven't further investigated.

\subsubsection{Shift correction.} The first step of our algorithm is to compute all shifts $s_{\alpha_i}, i=0,\ldots,P-1$. We use the same approach~\cite{desbat14} as for functions, but for distributions of compact support.  It's based on the well-known property that the projection of the center of mass of a function $f$ is the center of mass of its projection. This property is also valid for distributions of compact support: let $f \in \mathscr{E}'_2$ and define the mass of  $f$ by  $\langle f, \one \rangle$ (where $\one$ is the constant function equal to 1) and $\vec{c}_f \eqdef \frac{\langle f, \vec{x} \rangle}{\langle f, \one \rangle}$ the center of mass of $f$. For all $\alpha\in[0,2\pi)$,  the mass of the projection $\mathcal{R}_{\alpha}f \in \mathscr{E}'_1$ is the mass of $f$: 
\begin{eqnarray}
\left(\mathcal{R}_{\alpha}f,\one\right)
={\langle f, \mathcal{R}_{\alpha}^{*}\one \rangle} 
={\langle f, \one \rangle}, 
\end{eqnarray}
the center of mass of the projection is the projection of the center of mass:
\begin{eqnarray}
c_{\mathcal{R}_{\alpha}f}
\eqdef\frac{\left(\mathcal{R}_{\alpha}f(s),s\right)}{\left(\mathcal{R}_{\alpha}f,\one\right)}
=
\frac{\langle f, \vec{x} \cdot \vec{\theta}_{\alpha}\rangle}{\langle f, \one \rangle} 
= \vec{c}_f \cdot  \vec{\theta}_{\alpha}.
\end{eqnarray}

We suppose that we can identify in each projection the projections of a group g of $n_\groupeG$ markers, i.e., 
we can identify $v_{ij}$, $j\in G$, the projected center of the marker $j$ in the projection $i$ in the projection data 
\begin{eqnarray}
m^{\:\groupeG}_{i}(s) = \mathcal{R}_{\alpha_i}\sum_{j\in G} \delta_{\vec{c}_j}(s-s_{\alpha_i})
= \sum_{j\in G}\delta_{\vec{c}_j \cdot \vec{\theta}_{\alpha_i}}(s-s_{\alpha_i}),
\end{eqnarray}
where $G \subset \mathbb{N}$ is the index set of the markers of the group g.
The order $0$  moment of $m^{\:\groupeG}_{i}$ is the mass of the compact support distribution:
\begin{eqnarray}
\fl M^{\:\groupeG}_0(i)=(m^{\:\groupeG}_{i}(s), 1)
=\left(\sum_{j\in G} \delta_{\vec{c}_j \cdot \vec{\theta}_{\alpha_i}}(s-s_{\alpha_i}), 1\right)
=\left(\sum_{j\in G} \delta_{\vec{c}_j \cdot \vec{\theta}_{\alpha_i}}(s), 1\right)
\nonumber\\  =n_\groupeG,
\end{eqnarray}
the order $1$ moments of $m^{\:\groupeG}_{i}$:
\begin{eqnarray}
\fl M^{\:\groupeG}_1(i)=(m^{\:\groupeG}_{i}(s), s)= \left(\sum_{j\in G} \delta_{\vec{c}_j \cdot \vec{\theta}_{\alpha_i}}(s-s_{\alpha_i}), s\right)
=\left(\sum_{j\in G} \delta_{\vec{c}_j \cdot \vec{\theta}_{\alpha_i}}(s), s+s_{\alpha_i}\right) \nonumber\\
=\sum_{j\in G} \vec{c}_j \cdot \vec{\theta}_{\alpha_i}+s_{\alpha_i}n_\groupeG.
\label{eqn:moment_order_one_G_par}
\end{eqnarray}
From~\eref{eqn:moment_order_one_G_par} we have
\begin{eqnarray}
    s_{\alpha_i}=\frac{M^{\:\groupeG}_1(i)-\sum_{j\in G} \vec{c}_j \cdot \vec{\theta}_{\alpha_i}}{n_\groupeG}.
\end{eqnarray}
Thus, if we choose the center of mass of the group g as the center of the coordinate system, then $\sum_{j\in G} \vec{c}_j = \vec{0}$ and
\begin{eqnarray}
    s^{\:\groupeG}_{\alpha_i}=\frac{M^{\:\groupeG}_1(i)}{n_\groupeG}
    =\frac{\left(\sum_{j\in G} \delta_{v_{ij}}(s), s\right)}{n_\groupeG}
    =\frac{\sum_{j\in G}v_{ij}}{n_\groupeG}.
    \label{eqn:shiftcorrection}
\end{eqnarray}

If the marker set is either the "horizontal" ($l=\rmh$) or the "vertical" ($l=\rmv$) group, for the sake of simplicity, we will denote $ \vec{c}^{\:l}_j$ ($j=1,\ldots,n_l$, $n_l$ is either $\nh$ or $\nv$) the marker locations and $v^l_{ij}$ ($j=1,\ldots,n_l$) the detected locations of their corresponding projections, $i=0,\ldots,P-1$. If we choose the center of mass of the marker group $l$ as the center of the coordinate system ($l$ being either $\rmh$ or $\rmv$), then the shift corrections $s^l_{\alpha_i}$ and the corresponding estimators $\slalphaihat$ for each projection in each coordinate system are 
\begin{eqnarray}
    \eqalign{s^l_{\alpha_i}
    =\frac{M^l_1(i)-\sum_{j=1}^{n_l} \vec{c}^{\:l}_j \cdot \vec{\theta}_{\alpha_i}}{n_l}
    =\frac{M^l_1(i)}{n_l}
    =\frac{\left(\sum_{j=1}^{n_l} \delta_{v^l_{ij}}(s), s\right)}{n_l}, \cr 
    \slalphaihat = \frac{\sum_{j=1}^{n_l} v^l_{ij}}{n_l}.}
\label{eqn:slalphai}
\end{eqnarray}
 
By default, we work in the coordinate system centered at the center of mass of {\em all} markers, then the estimator $\salphaihat$ of $s_{\alpha_i}$ is given by
\begin{eqnarray}
\label{eqn:salphai}
    \salphaihat = \frac{M_1(i)}{n}=\frac{\sum_{j=1}^{n} v_{ij}}{n}, \: i=0,\ldots,P-1.
\end{eqnarray}

However, we compute three different possible shifts for each projection ($\slalphaihat, l\in \left\{\rmh,\rmv\right\}$, see~\eref{eqn:slalphai}, and $\salphaihat$, see~\eref{eqn:salphai}), one for each coordinate system center. With this, we can compensate for the shifts in projections in order to make the projection angle estimation formulas more simple. 
We separately  perform shift corrections by $s^{\rmh}_{\alpha_i}$ associated to the center of mass of the "horizontal" group of markers and by $s^{\rmv}_{\alpha_i}$ associated to the center of mass of the "vertical" group of markers in order to obtain simple formulas to identify $\cos{\alpha_i}$ and $\sin{\alpha_i}$ using only moments of order $2$ and $3$ of the data (once the shifts have been corrected by moments of order 0 and 1). In the next paragraphs, the angles $\alpha_i$, $i=0,\ldots,P-1$, are estimated relative to the "horizontal" line of markers (the line of small markers in~\Fref{fig:truncationradon}~(a) belonging to the group $\rmh$).

\subsubsection{Finding of angles.}
The data are defined for each of two groups $l=\rmh$ or $l=\rmv$ separately 
$m^l_{i}(s)=\mathcal{R}_{\alpha_i}\left(\sum_{j=1}^{n_l} \delta_{\vec{c}^{\:l}_j}\right)(s-s_{\alpha_i})$. The data are corrected twice with the consistent shifts $\slalphaihat$, $l\in\{\rmh,\rmv\}$, identified by~\eref{eqn:slalphai}, so 
\begin{eqnarray}
    \tilde{m}^l_{i}(s)=m^l_{i}(s+\slalphaihat) 
    = \sum_{j=1}^{n_l} \delta_{v^l_{ij}}(s+\slalphaihat)
    = \sum_{j=1}^{n_l} \delta_{v^l_{ij}-\slalphaihat}(s)
    = \sum_{j=1}^{n_l} \delta_{\tildev^l_{ij}}(s),
\end{eqnarray}
where $\tildev^l_{ij}=v^l_{ij}-\slalphaihat$. In the following, we present analytical formulas based on the order 2 and 3 moments  $M^l_2$ and  $M^l_3$ respectively of the shift corrected data  $\tilde{m}_i^l$.

The order $2$ moment conditions of the shift corrected measured data:  
\begin{eqnarray}
\fl M^l_2(i)=\left(\tilde{m}^l_{i}(s), s^2\right) 
=  \left(\mathcal{R}_{\alpha_i}\left(\sum_{j=1}^{n_l} \delta_{\vec{c}^{\:l}_j}\right), s^2\right)
=\left(\sum_{j=1}^{n_l} \delta_{\vec{c}^{\:l}_j \cdot \vec{\theta}_{\alpha_i}},s^2\right)=\sum_{j=1}^{n_l} (\vec{c}^{\:l}_j \cdot \vec{\theta}_{\alpha_i})^2 \nonumber\\ =\sum_{j=1}^{n_l}(c^l_{j1} \cos{\alpha_i}+c^l_{j2}\sin{\alpha_i})^2 \nonumber\\  
= a^l_{20}\cos^2{\alpha_i}+2a^l_{11}\cos{\alpha_i}\sin{\alpha_i}+a^l_{02}\sin^2{\alpha_i},
\end{eqnarray}
where $a^l_{20}=\sum_{j=1}^{n_l}(c^l_{j1})^2$, $a^l_{11}=\sum_{j=1}^{n_l} c^l_{j1}c^l_{j2}$, $a^l_{02}=\sum_{j=1}^{n_l}(c^l_{j2})^2$ are only depending on the marker center coordinates. The coefficients $a^l_{20}$ and $a^l_{02}$ are positive as soon as $n_l>1$. Since  $\vec{c}^{\:\rmh}_j=(c^\rmh_{j1},0)^T$ for the "horizontal" line of markers in the corresponding coordinate system  
centered at the center of mass of the $\nh$ markers on the "horizontal" line and $\vec{c}^{\:\rmv}_j=(0,c^\rmv_{j2})^T$ for the "vertical" line in the corresponding coordinate system centered at the center of mass of the $\nv$ markers on the "vertical" line, we have
\begin{eqnarray}
    M^\rmh_2(i)=a^\rmh_{20}\cos^2{\alpha_i}, \qquad M^\rmv_2(i)=a^\rmv_{02}\sin^2{\alpha_i}.
\label{eqn:m2radon}
\end{eqnarray}
The order 2 moments are computed from the projection data $\tilde{m}^l_{i}$: $M^l_2(i)=\left(\sum_{j=1}^{n_l} \delta_{\tildev^l_{ij}}(s), s^2\right)=\sum_{j=1}^{n_l} (\tildev^l_{ij})^2>0$.
Suppose that we know that $\alpha_0$ and $\alpha_1$ are different and that $M^\rmh_2(0) M^\rmv_2(1)-M^\rmh_2(1) M^\rmv_2(0) \neq 0$, see~\eref{eqn:alpha0}. 
We can define 4 unknowns $a^\rmh_{20}$, $a^\rmv_{02}$, $\alpha_0$, $\alpha_1$ from the system
\begin{eqnarray}
  \left\{\begin{array}{@{}l@{}l@{}l@{}}
    M^\rmh_2(0)=a^\rmh_{20}\cos^2{\alpha_0}\\
    M^\rmv_2(0)=a^\rmv_{02}\sin^2{\alpha_0}\\
    M^\rmh_2(1)=a^\rmh_{20}\cos^2{\alpha_1}\\
    M^\rmv_2(1)=a^\rmv_{02}\sin^2{\alpha_1}
  \end{array}\right. \Leftrightarrow   \left\{\begin{array}{@{}l@{}l@{}l@{}}
    M^\rmh_2(0)=\cos^2{\alpha_0}\cdot M^\rmh_2(1)/\cos^2{\alpha_1}\\
    M^\rmv_2(0)=\sin^2{\alpha_0}\cdot M^\rmv_2(1)/\sin^2{\alpha_1}\\
    a^\rmh_{20}=M^\rmh_2(1)/\cos^2{\alpha_1}\\
    a^\rmv_{02}=M^\rmv_2(1)/\sin^2{\alpha_1}
  \end{array}\right..
\end{eqnarray}
From the second equation 
\begin{eqnarray}
    \cos^2{\alpha_1}=1-\sin^2{\alpha_1}=1-\frac{M^\rmv_2(1)\sin^2{\alpha_0}}{M^\rmv_2(0)}= \frac{M^\rmv_2(0)-M^\rmv_2(1)\sin^2{\alpha_0}}{M^\rmv_2(0)},
\end{eqnarray}
then from the first equation
\begin{eqnarray}
\fl M^\rmh_2(0) \frac{M^\rmv_2(0)-M^\rmv_2(1)\sin^2{\alpha_0}}{M^\rmv_2(0)}=M^\rmh_2(1) \cos^2{\alpha_0} \Rightarrow \nonumber\\
M^\rmh_2(0) M^\rmv_2(0)-M^\rmh_2(0) M^\rmv_2(1)\sin^2{\alpha_0}=M^\rmh_2(1) M^\rmv_2(0) (1-\sin^2{\alpha_0}) \Rightarrow \nonumber\\
(M^\rmh_2(0) M^\rmv_2(1)-M^\rmh_2(1) M^\rmv_2(0)) \sin^2{\alpha_0}  =M^\rmh_2(0) M^\rmv_2(0) \nonumber\\ -M^\rmh_2(1) M^\rmv_2(0).
\end{eqnarray}

We can choose one solution $\alpha_0$ among 4 possibilities (from symmetry ambiguities, see~\eref{eqn:symmetry1} and~\eref{eqn:symmetry2}), 
e.g., we assume that $\alpha_0 \in (0,\pi/2)$, then $\alpha_0$ can be uniquely estimated from 
\begin{eqnarray}
\sin^2{\alphahatz}= \frac{M^\rmh_2(0) M^\rmv_2(0)-M^\rmh_2(1) M^\rmv_2(0)}{M^\rmh_2(0) M^\rmv_2(1)-M^\rmh_2(1) M^\rmv_2(0)}.
\label{eqn:alpha0}
\end{eqnarray}

We can calculate the coefficients $a^\rmh_{20}$ and $a^\rmv_{02}$ from
\begin{eqnarray}
    {\ahat}^\rmh_{20}=\frac{M^\rmh_2(0)}{\cos^2{\alphahatz}}, \qquad {\ahat}^\rmv_{02}=\frac{M^\rmv_2(0)}{\sin^2{\alphahatz}}.
\label{eqn:a20a02}
\end{eqnarray}

With similar computations as for the order 2 moments, we obtain for the order 3 moments:
\begin{eqnarray}
    M^\rmh_3(i)=a^\rmh_{30}\cos^3{\alpha_i}, \qquad M^\rmv_3(i)=a^\rmv_{03}\sin^3{\alpha_i},
\label{eqn:m3radon}
\end{eqnarray}
where $a^\rmh_{30}=\sum_{j=1}^{\nh}(c^\rmh_{j1})^3$, $a^\rmv_{03}=\sum_{j=1}^{\nv}(c^\rmv_{j2})^3$. We assume that $a^\rmh_{30}$ and $a^\rmv_{03}$ are non-zero (in practice, this can be detected by~\eref{eqn:m3radon} and avoided by a non-singular choice of the marker positions, singular distributions of markers yielding $a^\rmh_{30}=\sum_{j=1}^{n_\rmh}(c^\rmh_{j1})^3=0$ or  $a^\rmv_{03}=\sum_{j=1}^{n_\rmv}(c^\rmv_{j2})^3=0$ are highly unlikely, see also the subsection~\ref{sect:NumericalSimulationParallel}). Thus, 
\begin{eqnarray}
    \ahat^\rmh_{30}=\frac{M^\rmh_3(0)}{\cos^3{\alphahatz}}, \qquad \ahat^\rmv_{03}=\frac{M^\rmv_3(0)}{\sin^3{\alphahatz}}.
\label{eqn:a30a03}
\end{eqnarray}

Using \eref{eqn:m2radon} and~\eref{eqn:m3radon}, $\alpha_i$ can be estimated from 
\begin{eqnarray}
    \cos{\alphahati}=\frac{\ahat^\rmh_{20}M^\rmh_3(i)}{\ahat^\rmh_{30}M^\rmh_2(i)}, \qquad \sin{\alphahati}=\frac{\ahat^\rmv_{02}M^\rmv_3(i)}{\ahat^\rmv_{03}M^\rmv_2(i)}.
\label{eqn:alphai}
\end{eqnarray}
Here $M^l_3(i)$ can be computed from our data with $M^l_3(i)=\left(\sum_{j=1}^{n_l} \delta_{\tildev^l_{ij}}(s), s^3\right)=\sum_{j=1}^{n_l} (\tildev^l_{ij})^3$.
Thus, we can identify analytically the projection angles $\alpha_i$ from~\eref{eqn:alphai} based on~\eref{eqn:alpha0}, \eref{eqn:a20a02}, \eref{eqn:a30a03}. Then~\eref{eqn:salphai} and~\eref{eqn:alphai} give us the new algorithm. Note that our analytical algorithm is based only on the local information about the centers of markers, thus, an object of interest can be truncated. Moreover, we don't require the knowledge of any angles, but can compute all unknown acquisition parameters.

\subsection{Numerical simulations}
\label{sect:NumericalSimulationParallel}
All numerical experiments were performed with Python 3. 
The calibration marker set is presented within the object in~\Fref{fig:truncationradon}~(a) and simulated truncated projections in~\Fref{fig:truncationradon}~(b). The origin of the marker coordinate system is the center of mass of 6 Dirac distributions. The axes of the coordinate system attached to the marker set are the two perpendicular marker lines. In this coordinate system, the marker positions are in cm: 
\begin{itemize}
    \item $(-2.4,0)^T$, $(0.4,0)^T$, $(2.3,0)^T$ for the first line;
    \item $(-0.1,-2.5)^T$, $(-0.1,0.5)^T$, $(-0.1,2)^T$ for the second line.
\end{itemize}

For these calibration "vertical" and "horizontal" marker distributions $a^\rmh_{30}=\sum_{j=1}^{n_\rmh}(c^\rmh_{j1})^3=-4.95$ and  $a^\rmv_{03}=\sum_{j=1}^{n_\rmv}(c^\rmv_{j2})^3=-7.5$. They are non-zero as required in the previous section, see~\eref{eqn:alphai}. 
More generally, it can be easily checked that if the moment of order 1 of three aligned Dirac distributions is zero (thus, we work in their center of mass), then the moment of order 3 (in the coordinate system centered at the center of mass of this line) can only be zero if these three Dirac distributions are equidistant (which is highly unlikely for random distributions).
Note that $c^\rmh_{j1}$ and $c^\rmv_{j2}$ are coordinates on the respective lines with the center of each coordinate system in the center of mass of each respective group of markers.  
The estimation error of $a^\rmh_{30}$ and $a^\rmv_{03}$ from~\eref{eqn:a30a03} in the case of noise-free data is of order $10^{-15}$.

We simulated $P=80$ random projection angles  $\alpha_i$, $i=0,\ldots, P-1$, in $[0,\pi]$ as true values (in the marker coordinate system),  half of them is in $(0,\pi/2)$, the second half is in $(\pi/2,\pi)$ (in order to avoid $0$, $\pi/2$ and $\pi$ for which projected marker centers overlap). We assume that the projections of the set of Dirac distributions $m_{i}(s)$  at angle $\alpha_i$, see~\eref{eqn:misvij} and~\eref{eqn:mis}, can be used to extract $v_{ij}$. The noise-free projections $v_{ij}$ were defined by~\eref{eqn:vijdef}, where the shifts $s_{\alpha_i}$ were randomly generated according to the uniform distribution on  $[-0.05, 0.05)$. The inputs of our algorithm are two arrays of the size $(3,P)$. The first array consists of all $P$ projections $v^\rmh_{ij}$, $i=0,\ldots,P-1$, $j=1,\ldots,3$, of $3$ markers from the first "horizontal" marker line. The second array consists of $v^\rmv_{ij}$, $i=0,\ldots,P-1$, $j=1,\ldots,3$, from the "vertical" marker line.

\begin{table}
\caption{\label{tabparallel} Mean absolute errors for shifts ($\ErrS$, see~\eref{eqn:ErrS}) and angles ($\ErrAone$ and $\ErrAtwo$, see~\eref{eqn:ErrAone} and~\eref{eqn:ErrAtwo}) for numerical experiments with non-noisy and noisy projections.} 

\begin{indented}
\lineup
\item[]\begin{tabular}{lp{2cm}p{2cm}p{2.5cm}p{2.5cm}}
\br                              
Noise level&Noise std in cm&Error for shifts in cm& Error for angles $\alphaihatone$  in rad  for $\alphazhatone \in (0, \pi/2)$&Error for angles $\alphaihattwo$ in rad for $\alphazhattwo \in (\pi/2, \pi)$\cr 
~ & ~ & $\ErrS$ &  $\ErrAone$ &  $\ErrAtwo$ \cr
\mr
      0\%&0&$5.84\times10^{-17}$&$3.71\times10^{-16}$&$3.97\times10^{-16}$\cr
        10\%&0.001&$3.26\times10^{-4}$&$2.21\times10^{-3}$&$2.21\times10^{-3}$\cr
        50\% &0.005&$1.61\times10^{-3}$&$1.13\times10^{-2}$&$1.28\times10^{-2}$\cr 
        100\% &0.01&$3.22\times10^{-3}$&$2.45\times10^{-2}$&$3.54\times10^{-2}$\cr
        200\% &0.02&$6.51\times10^{-3}$&$6.10\times10^{-2}$&$8.20\times10^{-2}$\cr
\br
\end{tabular}
\end{indented}
\end{table}

The estimator $\alphaihatone$ is given with the choice $\alphazhatone\in (0,\pi/2)$ in~\eref{eqn:alpha0}. 
According to~\eref{eqn:alpha0}, we have 4 possible solutions for the angle $\alpha_0$, thus 4 possible angle sets compatible with the initial data. To illustrate this ambiguity, we present the results for another solution of~\eref{eqn:alpha0} with $\alphazhattwo \in (\pi/2, \pi)$  yielding then, with same formulas~\eref{eqn:salphai} and~\eref{eqn:alphai}, the estimator $\alphaihattwo$ for $i=1,\ldots,P-1$ of $\pi-\alpha_i$. According to~\eref{eqn:symmetry2}, these solutions correspond to markers differing by the symmetry $\sigma_2$ according to the $x_2$-axis. We computed the errors $\ErrS \eqdef 1/P\sum_{i=0}^{P-1}\abs{ \salphaihat - s_{\alpha_i}}$, 
$\ErrAone \eqdef 1/P\sum_{i=0}^{P-1}\abs{ \alphaihatone - \alpha_i}$ and
$\ErrAtwo \eqdef {1}/{P}\sum_{i=0}^{P-1} \abs{\alphaihattwo - (\pi-\alpha_i)}$ (also given in~\eref{eqn:ErrS}, ~\eref{eqn:ErrAone} and~\eref{eqn:ErrAtwo} with $\nbk=1$), see~\Tref{tabparallel},  first row (0\%), third to fifth columns.
From the exact projection values $v^l_{ij} = \vec{c}^{\:l}_j \cdot \vec{\theta}_{\alpha_i}+s_{\alpha_i}$, $i=0,\ldots, P-1$, $j=1,\ldots,3$, $l\in\{\rmh,\rmv\}$, the estimations $\salphaihat$, $\alphaihatone$ and $\alphaihattwo$ are almost perfect. 

In practice, the detection of $v^l_{ij}$ contains errors. We modeled detection errors with a Gaussian noise $N(0,\sigma)$ added to $v^l_{ij}$, where $\sigma$ is equal to the product of the noise level and the pixel size of the detector $0.01$ cm. In rows 2 (noise level of $10\%$) to 5 (noise level of $200\%$) of~\Tref{tabparallel} we present results of noisy experiments. We computed mean absolute errors from $\nbk = 100$ realizations of the Gaussian noise, i.e., from $v^l_{ijk}=v^l_{ij}+\epsilon_{ijk}$, $i=0,\ldots, P-1$, $j=1,\ldots,3$, $l\in\{\rmh,\rmv\}$,  $k=1,\ldots, \nbk$, where  $\epsilon_{ijk}$ is a random value given by  $N(0,\sigma)$. 
For the shifts the computed error is given by 
\begin{eqnarray}
\ErrS \eqdef 
\frac{1}{\nbk}\sum_{k=1}^{\nbk}
\frac{1}{P}\sum_{i=0}^{P-1} |\hat{s}_{\alpha_{i}, k} - s_{\alpha_{i}}|, 
\label{eqn:ErrS}
\end{eqnarray}
where $s_{\alpha_{i}}$ is the true shift of the projection $i$, $\hat{s}_{\alpha_{i}, k}$ is its estimation with~\eref{eqn:salphai} from the data  $v^l_{ijk}$, $j=1,\ldots,3$, $l\in\{\rmh,\rmv\}$.
For projection angles the  computed errors are given by 
\begin{eqnarray}
\ErrAone \eqdef 
\frac{1}{\nbk}\sum_{k=1}^{\nbk}
\frac{1}{P}\sum_{i=0}^{P-1} \abs{\alphahatone_{i,k} - \alpha_{i}}
\mbox{ and } 
\label{eqn:ErrAone}
\end{eqnarray}
\begin{eqnarray}
\ErrAtwo \eqdef 
\frac{1}{\nbk}\sum_{k=1}^{\nbk}
\frac{1}{P}\sum_{i=0}^{P-1} \abs{\alphahattwo_{i,k} - (\pi-\alpha_{i})},
\label{eqn:ErrAtwo}
\end{eqnarray}
where $\alpha_{i}$ is the true angle of the projection $i$, $\alphahatone_{i,k}$ and $\alphahattwo_{i,k}$ are the estimations with~\eref{eqn:alphai} from the data  $v^l_{ijk}$, $j=1,\ldots,3$, $l\in\{\rmh,\rmv\}$.
 We see in~\Tref{tabparallel} that the errors are essentially proportional to the noise level.

We repeated our experiments with similar calibration sets (two perpendicular lines of markers) of various sizes. We observed the well-known result: the angle estimation error is essentially a decreasing function of the size of the marker set. The shift estimation error is not sensitive to this size.

\section{\redC{Range conditions for the fan-beam transform on distributions and geometric calibration} }
\label{SecSelfCalibFanBeam}
\subsection{Mathematical results}

Let us consider $\mathcal{D}_{\lambda}f(y)\eqdef \mathcal{D}f(\lambda,y)$, where $\mathcal{D}$ is defined in~\eref{eqn:fanbeamdef}, $\lambda\in\bR$, $y\in\bR$,  for $f \in \mathscr{D}_2$ with support in $X_2=(D_1,D_2) \times \mathbb{R}$, $0<D_1<D_2<D$, see~\Fref{fig:geometries}~(b). 
Thus, $\mathcal{D}_{\lambda}f$ is a function of one variable $y\in\bR$. The operator $\mathcal{D}_{\lambda}$ is generalized to distributions through its dual operator $\mathcal{D}_{\lambda}^{*}$: 
\begin{eqnarray}
    \fl (\mathcal{D}_{\lambda}f,\phi)=\int_{-\infty}^{ +\infty} \mathcal{D}_{\lambda}f(y) \phi(y) \rmd y = \int_{-\infty}^{+\infty} \int_{0}^{+\infty} f(D-lD, \lambda+ly-l\lambda)\rmd l \phi(y) \rmd y \nonumber\\
    = \frac{1}{D}\int_{-\infty}^{+\infty} \int_{-\infty}^{D} f \left( u,\frac{u\lambda}{D}+\frac{(D-u)y}{D}\right) \rmd u \phi(y) \rmd y \nonumber\\=\int_{-\infty}^{D} \int_{-\infty}^{+\infty} f(u,v) \phi \left( \frac{vD-u\lambda}{D-u} \right) \frac{1}{D-u} \rmd v \rmd u \nonumber\\=\int_{D_1}^{D_2} \int_{-\infty}^{+\infty} f(u,v) \phi \left( \frac{vD-u\lambda}{D-u} \right) \frac{1}{D-u} \rmd v \rmd u=\langle f,\mathcal{D}_{\lambda}^{*} \phi \rangle, 
\end{eqnarray}
where we first made the change of variable $u=D-lD$, $\rmd l=-\rmd u/D$, $\lambda+l(y-\lambda)=\frac{u\lambda}{D}+\frac{(D-u)y}{D}$ and then the change of variable $v=\frac{u\lambda}{D}+\frac{(D-u)y}{D}$, $\rmd v=\frac{D-u}{D}\rmd y$. We denote by  
$(\cdot,\cdot)$ the canonical duality pairing for test functions on $\mathbb{R}$ and by $\langle \cdot,\cdot \rangle$ the canonical duality pairing for test functions on $X_2$. Thus, here we have
\begin{eqnarray}
    \forall \phi \in \mathscr{E}_1,  \: \mathcal{D}_{\lambda}^{*} \phi(\vec{x})= \frac{1}{D-x_1} \phi \left( \frac{x_2D-x_1\lambda}{D-x_1} \right).
   \label{eqn:adjointfanbeam}
\end{eqnarray}
We have $\mathcal{D}_{\lambda}^{*} \phi \in \mathscr{E}(X_2)$ as a composition of two smooth functions  $\vec{x}\mapsto \frac{x_2D-x_1\lambda}{D-x_1}$ and $\phi$ multiplied by the smooth function $\vec{x}\mapsto \frac{1}{D-x_1}$. 
\begin{definition}
 The fan-beam transform on a line $\mathcal{D}_{\lambda} f$ of $f \in \mathscr{E}'(X_2)$ is a distribution acting on the space $\mathscr{E}_1$ according to 
 \begin{eqnarray}
     \forall \phi \in \mathscr{E}_1, \: (\mathcal{D}_{\lambda}f,\phi) \eqdef \langle f,\mathcal{D}_{\lambda}^{*} \phi \rangle.
     \label{eqn:dualityfanbeam}
 \end{eqnarray}
\end{definition}

The linearity of the operator $\mathcal{D}_{\lambda} f$ is obvious. The continuity is equivalent to its boundedness. We give the proof of its boundedness in Appendix~\ref{SecAppendix}, Lemma~\ref{LemmaFanbeamBounded}. It's the proof of the first point~{\em(\ref{item_one_of_FanbeamRangeCondTh})} in the next theorem.

\begin{theorem}[Necessary range conditions for $\mathcal{D}_{\lambda}$ on distributions, $\mathbf{\lambda\in\mathbb{R}}$]
Let $g_{\lambda}\eqdef\mathcal{D}_{\lambda}f$ be the fan-beam transform on a line of $f\in \mathscr{E}'(X_2)$ for $\lambda\in\bR$, then:
\begin{enumerate}
    \item\label{item_one_of_FanbeamRangeCondTh} $g_{\lambda} \in  \mathscr{E}'_1$,
    \item\label{item_two_of_FanbeamRangeCondTh} for $k\in\bN$, we have the moment conditions:
    \begin{eqnarray}
        \left(g_{\lambda} (y), y^k\right)=\mathscr{P}_k(\lambda), \: \forall\lambda\in\bR,
    \label{eqn:dccfanbeamourmomentcond}
    \end{eqnarray}
    where $\mathscr{P}_k(\lambda)$ is a polynomial of degree at most $k$ in $\lambda$.
\end{enumerate}
\end{theorem}

\begin{proof}
The point~{\em(\ref{item_one_of_FanbeamRangeCondTh})} is proven in  Appendix~\ref{SecAppendix}, Lemma~\ref{LemmaFanbeamBounded}. 
For~{\em(\ref{item_two_of_FanbeamRangeCondTh})}, $\forall k\in\mathbb{N}$, the function $y\in\bR \mapsto y^k \in\bR$ is in $\mathscr{E}_1$:
\begin{eqnarray}
 \fl (\mathcal{D}_{\lambda}f(y), y^k)=\Braket{ f(\vec{x}),\mathcal{D}_{\lambda}^{*} (y^k)(\vec{x}) }=\Braket{ f(\vec{x}), \left( \frac{x_2D-x_1\lambda}{D-x_1} \right)^k \frac{1}{D-x_1} }   \nonumber\\= \Braket{ f(\vec{x}), \frac{1}{(D-x_1)^{k+1}} \sum_{i=0}^{k}\binom{k}{i} (x_2D)^{k-i} (-x_1\lambda)^i }\nonumber\\ 
 = \sum_{i=0}^{k}\binom{k}{i} \Braket{ f(\vec{x}), \frac{(x_2D)^{k-i}(-x_1)^i}{(D-x_1)^{k+1}} } \lambda^i=\mathscr{P}_k(\lambda).
\end{eqnarray}
\end{proof}

The fan-beam transform at fixed $\lambda$ of $\delta_{\vec{c}}$ (modeling a marker at $\vec{c}\in\bR^2$)  is
\begin{eqnarray}
 \fl \left(\mathcal{D}_{\lambda}\delta_{\vec{c}}(y), \phi(y)\right)=\Braket{ \delta_{\vec{c}}(\vec{x}) ,\mathcal{D}_{\lambda}^{*} \phi(\vec{x}) }=\Braket{ \delta_{\vec{c}} (\vec{x}), \phi \left( \frac{x_2D-x_1\lambda}{D-x_1} \right) \frac{1}{D-x_1}  } \nonumber\\
 =  \phi \left( \frac{c_2D-c_1\lambda}{D-c_1} \right) \frac{1}{D-c_1}
 = \frac{1}{D-c_1} \delta_{\tilde{c}} (\phi),  \qquad \tilde{c}=\frac{c_2D-c_1\lambda}{D-c_1}.
 \label{eqn:fanbeamofdirac}
\end{eqnarray}
Note that $\tilde{c}$ is the geometric projection of $\vec{c}$ in fan-beam geometry of $\mathcal{D}_{\lambda}$ (see~\cite{desbat19}). \redC{This theorem represents the first attempt to construct DCCs on distributions for the fan-beam geometry. In the following, we will demonstrate that this theoretical result can prove to be valuable for constructing a calibration procedure, similar to the approach used for the parallel geometry.}

\subsection{Geometric calibration problem to solve}
\label{SectSelfCalibFBoaL}

Existing geometric calibration methods for this fan-beam geometry  based on DCCs require again non-truncated data (see, for example,~\cite{nguyen20}). 
In~\Fref{fig:truncationfanbeam}, we present the example used in our numerical simulations, see subsection~\ref{sect:NumericalSimulationFanBeam}. The object of interest is the union of big disks and isn't in the field-of-view of any projection. 
As it can be seen in~\Fref{fig:truncationfanbeam}~(b), all projections are truncated and non-zero projection data are missing to compute moments. 
However, all marker projections are contained in all projection data. Thus, we will use only the local data concerning markers to check the range conditions~\eref{eqn:dccfanbeamourmomentcond}.

\begin{figure}[btph]
     \centering
     \begin{subfigure}[b]{0.48\textwidth}
         \centering
         \includegraphics[width=\textwidth]{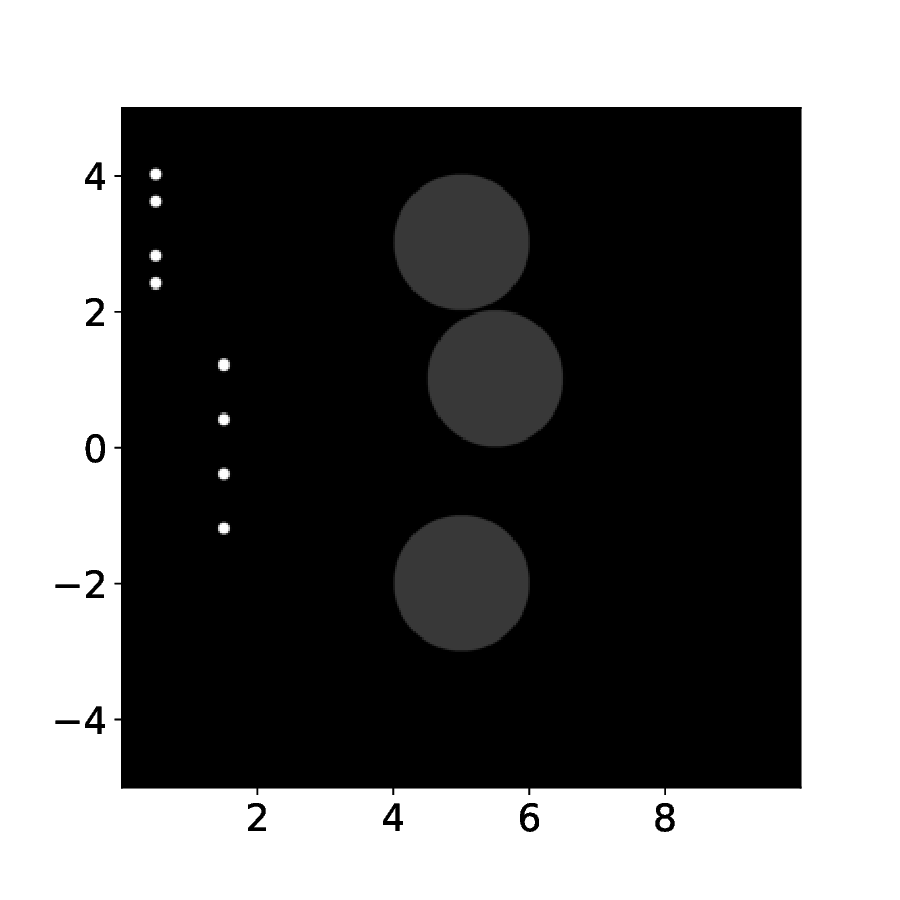}
         \caption{}
     \end{subfigure}
     \hfill
     \begin{subfigure}[b]{0.48\textwidth}
         \centering
         \includegraphics[width=\textwidth]{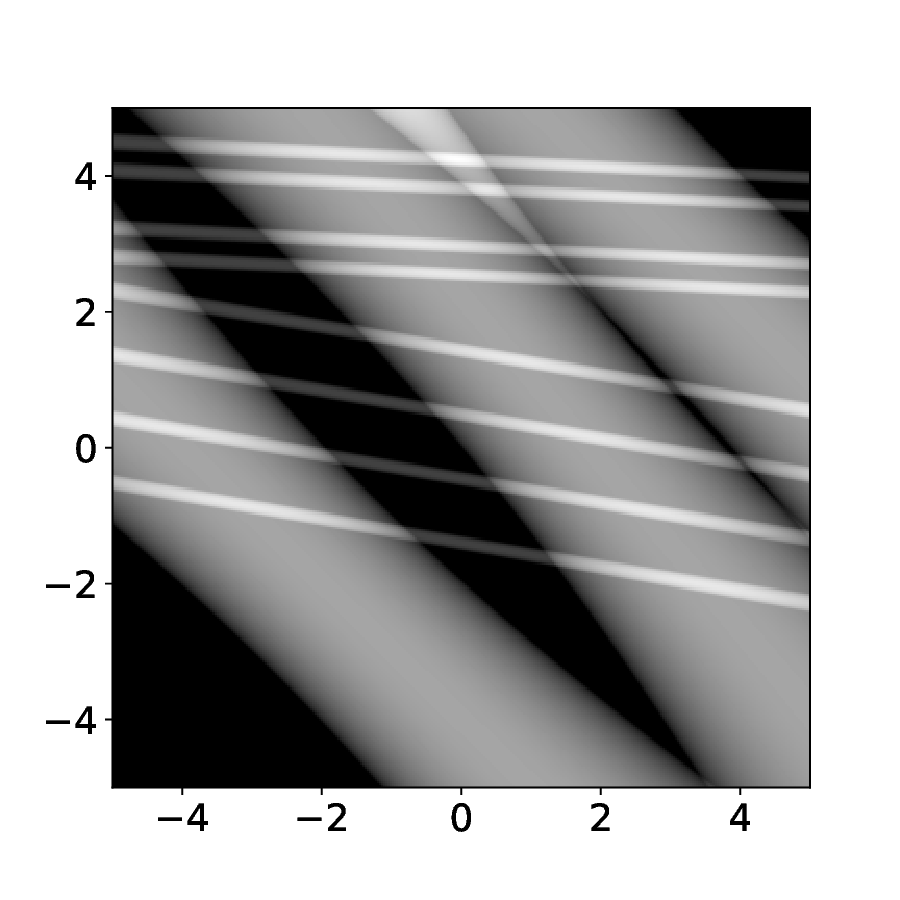}
         \caption{}
     \end{subfigure}
        \caption{(a) measured data $\fob+\fballs$: object  $\fob$ (the union of big disks) with markers $\fballs$ ($8$ small disks, i.e., two objects of $4$ markers) in $\mathbb{R}^2$; the horizontal axis is $x_1$, the vertical axis is $x_2$, the detector line is at $x_1=0$, the source line is at $x_1=D=10$. (b) Their {\em truncated} fan-beam projections; the horizontal axis is $\lambda$, the vertical axis is $y$. At fixed $\lambda_i\in [-5,5]$, $i=0,\ldots,P-1$, each projection $\mathcal{D}_{\lambda_i}(\fob+\fballs)$ {\em is truncated}. But we can remark that the marker set projections $\mathcal{D}_{\lambda_i}(\fballs)$ $\forall \lambda_i$ are {\em not truncated}.
        }
        \label{fig:truncationfanbeam}
\end{figure}

We assume that the source positions are unknown and that the detector could have been shifted at each projection (because of vibrations or because the detector is moving like in the geometry considered in~\cite{nguyen20}). 
We assume that we have $P$ projections 
\begin{eqnarray}
    \overline{m}_{i}(y)=\mathcal{D}_{\lambda_i}(\fob+\fballs)\left(y-y_{\lambda_i}\right), \: y\in\bR, \: i=0,\ldots,P-1,
\end{eqnarray}
where both source positions $\lambda_i\in\bR$ and shifts $y_{\lambda_i}\in\bR$, $i=0,\ldots,P-1$, are unknown. As previously, $\fob$ is the object attenuation function and $\fballs$ is the marker attenuation function. All projections $\mathcal{D}_{\lambda_i}\fob$ are truncated, but none of $\mathcal{D}_{\lambda_i}\fballs$. 
Thus, we focus only on $m_{i}(y)=\mathcal{D}_{\lambda_i} \left(\sum_{j=1}^{n} \delta_{\vec{c}_j}\right) \left(y-y_{\lambda_i}\right)$ obtained from $\mathcal{D}_{\lambda_i}\fballs\left(y-y_{\lambda_i}\right)$, where $\vec{c}_j$ is the true 2D center of the marker $j$, $j=1,\ldots,n$. We have $\mathcal{D}_{\lambda_i}\left(\sum_{j=1}^{n} \delta_{\vec{c}_j}\right)=\sum_{j=1}^{n} \frac{1}{D-c_{j1}}\delta_{\tilde{c}(\vec{c}_j,{\lambda_i},D)}$, where $\tilde{c}(\vec{c}_j,{\lambda_i},D)$ is the projection of $\vec{c}_j$ from the source position $\lambda_i$ as defined by~\eref{eqn:fanbeamofdirac}. Note that the detection of  $\tilde{c}(\vec{c}_j,{\lambda_i},D)$ in the projection data doesn't mean that we know $m_{i}(y)$, unless we know the source-detector distance $D$ and the marker $x_1$-coordinates $c_{j1}$, $j=1,\ldots,n$.

In the next subsection, we show that the geometric parameters $\lambda_i$, $y_{\lambda_i}$, $i=0,\ldots,P-1$, and the partially unknown 2D marker positions $\vec{c}_j$, $j=1,\ldots,n$, can be identified from the detection of the marker center projections and some properties of our designed calibration marker set.

\subsection{Calibration method}

\label{SectFanBeamNonUniqueSol}
Let us remind first that the geometric calibration based only on the measured data doesn't have a unique solution for functions.  Let $f$ be a smooth function of compact support  in $X_2$, we have (see~\cite{desbat19})
\begin{eqnarray}
\label{Eq:ShearAndTransForFunction}
\mathcal{D}f_{M,\vec{t}} (\lambda,y)
 =\mathcal{D}f(\lambda+\lambda',y-y')
\end{eqnarray}
with $f_{M,\vec{t}}(\vec{x}) \eqdef f\left(M\vec{x}+\vec{t}\right)$, the shearing matrix $M=\left( \matrix{ 1 & 0 \cr
(y'+\lambda')/D & 1} \right)$  and the translation $\vec{t}=(0,-y')^T\in\bR^2$. Thus, we can not identify from only range conditions the geometric parameters $\lambda$ and $y$ better than up to a global shift $\lambda'$ on $\lambda$ and a global shift $y'$ on $y$ (both origins of the source positions and of the detector positions of consistent data can be freely chosen). 

We have the same result for Dirac sums. 
 Let us associate to $\delta_{\vec{c}} \in \mathscr{E}'(X_2)$, $\vec{c}\in X_2$, the distribution $\deltacMt$ 
 defined by $\Braket{\deltacMt(\vec{x}), \phi(\vec{x}) }=\Braket{ \delta_{\vec{c}}(\vec{x}), \phi\left(M^{-1}(\vec{x}-\vec{t})\right) }$,  
 where $M^{-1}=\left( \matrix{ 1 & 0 \cr
-(y'+\lambda')/D & 1} \right)$, thus, $M^{-1}(\vec{x}-\vec{t})=\left( \matrix{ x_1 \cr
-(y'+\lambda')x_1/D+x_2+y' } \right)$. 
As in~\eref{eqn:distributionRT}, $(\delta_{\vec{c}})_{M,\vec{t}}$ is the distribution $\delta_{M^{-1}(\vec{c}-\vec{t})} \in \mathscr{E}'(X_2)$. Thus, in the geometric calibration with a Dirac distribution (or a sum of Dirac distributions by linearity of $\mathcal{D}_{\lambda}$) we have the similar ambiguity:
\begin{eqnarray}
\label{Eq:ShearAndTrans}
\fl \left(\mathcal{D}_{\lambda}(\delta_{\vec{c}})_{M,\vec{t}}(y), \phi(y)\right)= \Braket{ (\delta_{\vec{c}})_{M,\vec{t}}(\vec{x}), \phi \left( \frac{x_2D-x_1\lambda}{D-x_1} \right) \frac{1}{D-x_1} }\nonumber\\=\Braket{ \delta_{M^{-1}(\vec{c}-\vec{t})}(\vec{x}), \phi \left( \frac{x_2D-x_1\lambda}{D-x_1} \right) \frac{1}{D-x_1} }\nonumber\\= \frac{1}{D-c_1} \phi \left( \frac{(-(y'+\lambda')c_1/D+c_2+y')D-c_1\lambda}{D-c_1} \right)\nonumber\\ =\frac{1}{D-c_1} \phi \left( \frac{c_2D-c_1(\lambda+\lambda')}{D-c_1} +y' \right) = \left(\mathcal{D}_{\lambda+\lambda'} \delta_{\vec{c}}(y), \phi(y+y')\right) \nonumber\\ =\left(\mathcal{D}_{\lambda+\lambda'} \delta_{\vec{c}}(y-y'), \phi(y)\right).
\end{eqnarray}

The geometric calibration as defined in the  subsection~\ref{SectSelfCalibFBoaL} doesn't have a unique solution. After a global shift $\lambda'$ of the source position and a global shift $y'$ of the detector position, the data are still consistent.  Thus, we can choose the origin of the source position to be $\lambda_0=0$ and for the detector position $y_{\lambda_0}=0$. We want to find only one solution among all possible solutions. This solution will correspond to another object obtained from the initial one by a transformation including a shearing and a translation.
 
The scaling ambiguity in geometric self-calibration known from computer vision~\cite{pollefeys99} was shown for the X-ray 3D CB transform in~\cite{konik21}.
A similar scaling ambiguity occurs in our 2D divergent geometry. Without any additional information, we can't uniquely estimate $D$. 
Let $k\in\bR$, $k>0$, the fan-beam data $\mathcal{D}f(\lambda,y)=\int_0^{+ \infty} f(kD-lkD, \lambda+ly-l\lambda) \rmd l$ for the function $f(x_1,x_2)$ and the source-detector distance $kD$ are equal to the fan-beam data $\mathcal{D}\tilde{f}(\lambda,y)=\int_0^{+ \infty} f(k(D-lD), \lambda+ly-l\lambda)\rmd l$ for the function $\tilde{f}(x_1,x_2)=f(kx_1,x_2)$ and the source-detector distance $D$. The same is true for Dirac distributions: the available 2D projection $\frac{c_2 D - c_1\lambda}{D-c_1}$ of $(c_1,c_2)^T$ in the geometry with $D$  is equal to $\frac{c_2D/k-c_1\lambda/k}{D/k-c_1/k}$ of $(c_1/k,c_2)^T$ in the geometry with $D/k$. Thus, we can not estimate $D$ in self-calibration without any additional information. We assume that $D$ is known.

We now design a specific calibration marker set. We suppose that we have two (vertical) lines parallel to the detector and source lines with two different groups of $4$ markers as in~\Fref{fig:truncationfanbeam} (a). We use the index $l=\lone$ for the first line, $l=\ltwo$ for the second line. As usual, both parallel calibration marker sets are fixed to the object (have the same movement if any). The positions of both marker sets are assumed to be unknown.  We also assume that the detector may jitter during the object translation. In practice, putting both lines of marker sets as close as possible to the detector line reduces the possibility of truncation in marker projections. 

We do not consider in this work the marker detection problem in projections. However, we remind that a popular detection trick in computer vision is the cross-ratio invariance (see, for example, \cite{park00} or \cite{aichert18}). The cross-ratio of 4 collinear points is invariant by the divergent projection. Thus, in our marker set, each of both aligned sets of 4 markers has its own specific cross-ratio.  We can classify in all projections 8 markers into two groups of 4, each belongs to its specific line. In real applications, the markers are not Dirac distributions, but balls of finite radius. In 3D CB applications, it's known that the centroid of the projection of a ball is not the projection of the center of the ball, see~\cite{Desbat2006ConeBeamImagingOfDelta}. But it's also known that this error have in general negligible effects on geometric calibration~\cite{Desbat2006ConeBeamImagingOfDelta}. In our 2D fan-beam geometry, it's known that the center of mass of the projection of a disk indicator is the projection of the center of mass of this disk indicator affected by a weight $1/(D-x_1)^2$, see~\cite{desbat19}. We haven't investigated this point, but as in CB, we assume that this effect is in general negligible for the geometric calibration.

We assume that the first 4 marker centers for the first line have coordinates $\vec{c}^{\:\lone}_1=(C_{\lone},p_{\lone}-k_1L)^T$, $\vec{c}^{\:\lone}_2=(C_{\lone},p_{\lone}-L)^T$, $\vec{c}^{\:\lone}_3=(C_{\lone},p_{\lone}+L)^T$, $\vec{c}^{\:\lone}_4=(C_{\lone},p_{\lone}+k_1L)^T$, where $p_{\lone}\in\bR$ and the position $C_{\lone}$ ($0<D_1<C_{\lone}<D_2<D$) of the line $x_1=C_{\lone}$ are unknown, but we assume to know $L$ and $k_1$ ($L>0$, $k_1>0$). The centers for the second line are $\vec{c}^{\:\ltwo}_1=(C_{\ltwo},p_{\ltwo}-k_2L)^T$, $\vec{c}^{\:\ltwo}_2=(C_{\ltwo},p_{\ltwo}-k_3L)^T$, $\vec{c}^{\:\ltwo}_3=(C_{\ltwo},p_{\ltwo}+k_3L)^T$, $\vec{c}^{\:\ltwo}_4=(C_{\ltwo},p_{\ltwo}+k_2L)^T$, where $p_{\ltwo}\in\bR$ and the position $C_{\ltwo}$ ($0<D_1<C_{\ltwo}<D_2<D$) of the line $x_1=C_{\ltwo}$ are unknown, but we know  
$k_2>0$, $k_3>0$. To summarize, we know the specific pattern of each of both 4-marker sets, but we don't know their positions given by $(C_{\lone},p_{\lone})$ and $(C_{\ltwo},p_{\ltwo})$.

For each 4-marker set, $l=\lone$ or $l=\ltwo$, the data are defined by 
\begin{eqnarray}
    m^l_{i}(y)=\mathcal{D}_{\lambda_i}\ftol(y-y_{\lambda_i}), 
\end{eqnarray} 
where $\vec{c}^{\:l}_j\in X_2$ ($j=1,\ldots,4$, $l\in\{\lone,\ltwo\}$) are the marker locations.

The  order $0$ moments of projections $i=0,\ldots,P-1$ for each group $l\in\{\lone,\ltwo\}$:
\begin{eqnarray}
\fl M^l_0(i)=(m^l_{i}(y), 1)= \left(\mathcal{D}_{\lambda_i}\ftol(y-y_{\lambda_i}), 1\right)=\left( \sum_{j=1}^{4} \frac{1}{D-c^{\:l}_{j1}} \delta_{\tilde{c}(\vec{c}^{\:l}_j,{\lambda_i},D)}(y-y_{\lambda_i}), 1 \right) \nonumber\\ =\frac{4}{D-C_{l}}.
\end{eqnarray}

The order $1$ moments:
\begin{eqnarray}
\fl M^l_1(i)=\left(m^l_{i}(y), y\right)= \left(\mathcal{D}_{\lambda_i}\ftol(y-y_{\lambda_i}), y\right)=\left(\mathcal{D}_{\lambda_i}\ftol(y), y+y_{\lambda_i}\right) \nonumber\\
=\left( \sum_{j=1}^{4} \frac{1}{D-c^{\:l}_{j1}} \delta_{\tilde{c}(\vec{c}^{\:l}_j,{\lambda_i},D)}(y),y+y_{\lambda_i} \right)=\sum_{j=1}^{4} \frac{\tilde{c}(\vec{c}^{\:l}_j,{\lambda_i},D)}{D-c^{\:l}_{j1}}+\sum_{j=1}^{4} \frac{y_{\lambda_i}}{D-c^{\:l}_{j1}}\nonumber\\
=\sum_{j=1}^{4} \frac{c^{\:l}_{j2}D-c^{\:l}_{j1}\lambda_i}{(D-c^{\:l}_{j1})^2}+\sum_{j=1}^{4} \frac{y_{\lambda_i}}{D-c^{\:l}_{j1}}\nonumber\\
=-\lambda_i \frac{4C_{l}}{(D-C_{l})^2}+ y_{\lambda_i}M^l_0(i)+\sum_{j=1}^{4} \frac{c^{\:l}_{j2}D}{(D-C_{l})^2}.
\label{eqn:m1fanbeam}
\end{eqnarray}
According to~\eref{Eq:ShearAndTrans}, we can choose $\lambda_0=0$ and $y_{\lambda_0}=0$. Then for $i=0$
\begin{eqnarray}
\eqalign{ M^{l}_1(0)=\sum_{j=1}^{4} \frac{c^{\:l}_{j2}D}{(D-C_{l})^2}, \: l\in\{\lone,\ltwo\},\label{eqn:firstmomentzeroprojection} \cr
 \mbox{i.e.,~}
   M^{\lone}_1(0)=\sum_{j=1}^{4} \frac{c^{\:\lone}_{j2}D}{(D-C_{\lone})^2},  \: M^{\ltwo}_1(0)=\sum_{j=1}^{4} \frac{c^{\:\ltwo}_{j2}D}{(D-C_{\ltwo})^2}.}
\end{eqnarray}
We assume to have detected $v^l_{ij}$, the centers of the projected markers in projections, then 
\begin{eqnarray} 
M^l_1(i)=\left(\sum_{j=1}^{4}\frac{1}{D-c^{\:l}_{j1}} \delta_{v^{l}_{ij}}, y\right)
=\frac{1}{D-C_{l}}\sum_{j=1}^{4} v^{l}_{ij}.
\label{eqn:m1fanbeamtocompute}
\end{eqnarray}
If we put~\eref{eqn:firstmomentzeroprojection} and~\eref{eqn:m1fanbeamtocompute} in \eref{eqn:m1fanbeam} and multiple each side of the equation by $D-C_{l}$, we obtain
\begin{eqnarray} 
-\lambda_i \frac{4C_{l}}{D-C_{l}}+4y_{\lambda_i}=\sum_{j=1}^{4} v^{l}_{ij}-\sum_{j=1}^{4} v^{l}_{0j}.
\label{eqn:fanbeammain0}
\end{eqnarray}
Let us define new unknown variables $r_{\lone}>0$ and $r_{\ltwo}>0$ with
\begin{eqnarray}
   \eqalign{r_{l}= \frac{C_{l}}{D-C_{l}}, \: l\in\{\lone,\ltwo\}, \cr 
   \mbox{i.e.,~} r_{\lone}= \frac{C_{\lone}}{D-C_{\lone}}, \:
   r_{\ltwo}= \frac{C_{\ltwo}}{D-C_{\ltwo}},}
   \label{eqn:cLl}
\end{eqnarray}
then \eref{eqn:fanbeammain0} gives, $l\in\{\lone,\ltwo\}$, $i=1,\ldots,P-1$,
\begin{eqnarray} 
y_{\lambda_i}-\lambda_i r_l=\Delta \tilde{M}^l_1(i), \mbox{ where } \Delta \tilde{M}^l_1(i) =\frac{1}{4}\left( \sum_{j=1}^{4} v^{l}_{ij}-\sum_{j=1}^{4} v^{l}_{0j} \right),
 \label{eqn:fanbeammain1}
\end{eqnarray}
$\Delta \tilde{M}^l_1(i)$ can be computed from the projection data. We need more equations to find the $2P$ unknown $y_{\lambda_i}$, $\lambda_i$, $i=1,\ldots,P-1$ and $r_{\lone}$, $r_{\ltwo}$  than $2P-2$ equations in~\eref{eqn:fanbeammain1}.

The moment conditions of order $2$ (let us remind that $c^{\:l}_{j1}=C_{l}$, $j=1,\ldots,4$, $l\in\{\lone,\ltwo\}$) :
\begin{eqnarray}
\fl M^l_2(i)=\left(m^l_{i}(y), y^2\right)= \left(\mathcal{D}_{\lambda_i}\ftol(y), (y+y_{\lambda_i})^2\right)=\sum_{j=1}^{4} \frac{\left(\tilde{c}(\vec{c}^{\:l}_j,{\lambda_i},D)+y_{\lambda_i}\right)^2}{D-c^{\:l}_{j1}}\nonumber\\
=\frac{1}{{D-C_{l}}}\sum_{j=1}^{4} \left( \frac{c^{\:l}_{j2}D-c^{\:l}_{j1}\lambda_i}{D-C_{l}} +y_{\lambda_i}\right)^2\nonumber\\
=\frac{1}{{D-C_{l}}}\sum_{j=1}^{4} \left( \frac{c^{\:l}_{j2}D}{D-C_{l}} +y_{\lambda_i}-\lambda_i r_l \right)^2\nonumber\\ =\frac{1}{{D-C_{l}}}\sum_{j=1}^{4} \left( \frac{c^{\:l}_{j2}D}{D-C_{l}} +\Delta \tilde{M}^l_1(i) \right)^2.
\label{eqn:m2fanbeam}
\end{eqnarray}

From \eref{eqn:firstmomentzeroprojection}, \eref{eqn:m1fanbeamtocompute} and \eref{eqn:cLl} we have 
\begin{eqnarray}
  \sum_{j=1}^{4} \frac{c^{\:l}_{j2}D}{D-C_{l}}=\frac{D}{D-C_{l}} \sum_{j=1}^{4} c^{\:l}_{j2}=(1+r_l)\sum_{j=1}^{4} c^{\:l}_{j2}=\sum_{j=1}^{4} v^{l}_{0j},
  \label{eqn:fanbeamlong}
\end{eqnarray}
thus, from $\sum_{j=1}^{4} c^{\:l}_{j2}=4p_l$ and \eref{eqn:fanbeamlong}
\begin{eqnarray}
(1+r_l)4p_l=\sum_{j=1}^{4} v^{l}_{0j}.
\label{eqn:p1p2}
\end{eqnarray}

For each of our both marker sets $\lone$ and $\ltwo$ we have
\begin{eqnarray}
\sum_{j=1}^{4} (c^{\:\lone}_{j2})^2=4p^2_{\lone}+(2+2k^2_1)L^2 \mbox{ and }
\sum_{j=1}^{4} (c^{\:\ltwo}_{j2})^2=4p^2_{\ltwo}+(2k^2_2+2k^2_3)L^2,
\end{eqnarray}
then we can compute using~\eref{eqn:cLl}:
\begin{eqnarray}
\fl\sum_{j=1}^{4} \left( \frac{c^{\:l}_{j2}D}{D-C_{l}} \right)^2=(1+r_l)^2\sum_{j=1}^{4} (c^{\:l}_{j2})^2=\nonumber\\
=\cases{ (1+r_{\lone})^2\left(4p^2_{\lone}+(2+2k^2_1)L^2\right) & if $l=\lone$ \\ (1+r_{\ltwo})^2\left(4p^2_{\ltwo}+(2k^2_2+2k^2_3)L^2\right) & if $l=\ltwo$.}
\end{eqnarray}
From the detected $v^l_{ij}$ (the centers of the projected markers in projections) we have 
\begin{eqnarray} 
M^l_2(i)=\left(\sum_{j=1}^{4}\frac{1}{D-c^{\:l}_{j1}} \delta_{v^{l}_{ij}}, y^2\right)
=\frac{1}{D-C_{l}}\sum_{j=1}^{4} \left(v^{l}_{ij}\right)^2.
\label{eqn:m2fanbeamtocompute}
\end{eqnarray}
Now we can rewrite \eref{eqn:m2fanbeam} after multiplying each side of the equation by $D-C_{l}$:
\begin{eqnarray}
\fl \cases{ \sum_{j=1}^{4} (v^{\lone}_{ij})^2=(1+r_{\lone})^2\left(4p^2_{\lone}+(2+2k^2_1)L^2\right)+2\Delta \tilde{M}^{\lone}_1(i)\sum_{j=1}^{4} v^{\lone}_{0j}+4[\Delta \tilde{M}^{\lone}_1(i)]^2\\
\sum_{j=1}^{4} (v^{\ltwo}_{ij})^2=(1+r_{\ltwo})^2\left(4p^2_{\ltwo}+(2k^2_2+2k^2_3)L^2\right)+2\Delta \tilde{M}^{\ltwo}_1(i)\sum_{j=1}^{4} v^{\ltwo}_{0j}+4[\Delta \tilde{M}^{\ltwo}_1(i)]^2,} 
\end{eqnarray}
finally we obtain with \eref{eqn:p1p2} the following system:
\begin{eqnarray}
\fl \cases{ \sum_{j=1}^{4} (v^{\lone}_{ij})^2=(1+r_{\lone})^2(2+2k^2_1)L^2+\frac{1}{4}\left[\sum_{j=1}^{4} v^{\lone}_{0j} \right]^2+2\Delta \tilde{M}^{\lone}_1(i)\sum_{j=1}^{4} v^{\lone}_{0j}+4[\Delta \tilde{M}^{\lone}_1(i)]^2\\
\sum_{j=1}^{4} (v^{\ltwo}_{ij})^2=(1+r_{\ltwo})^2(2k^2_2+2k^2_3)L^2+\frac{1}{4}\left[\sum_{j=1}^{4} v^{\ltwo}_{0j} \right]^2+2\Delta \tilde{M}^{\ltwo}_1(i)\sum_{j=1}^{4} v^{\ltwo}_{0j}+4[\Delta \tilde{M}^{\ltwo}_1(i)]^2,}
\label{eqn:r1r2}
\end{eqnarray}
where $k_1$, $k_2$, $k_3$ and $L$ are known. 
We can compute uniquely $(1+r_{\lone})^2$ and $(1+r_{\ltwo})^2$ from these two equations~\eref{eqn:r1r2} and from one chosen projection $i \geq 1$.
Several projections could be used in order to improve the precision and the stability of the estimation of $r_l$, $l \in \{\lone,\ltwo\}$, in the case of noisy data (we used only one projection for our numerical simulations including experiments on noisy data). Suppose that the solution of the system is $R_l=(1+r_l)^2$, then $r_l=-1\pm R_l^{1/2}$. Since $r_l>0$, thus, from \eref{eqn:r1r2} we can uniquely compute $r_l=-1+R_l^{1/2}$, then from \eref{eqn:cLl} uniquely compute $C_{\lone}$ and $C_{\ltwo}$ for fixed $D$ and from \eref{eqn:p1p2} uniquely compute $p_{\lone}$ and $p_{\ltwo}$. For each $i \geq 1$ \eref{eqn:fanbeammain1} gives the system of two equations for computing the pair $y_{\lambda_i}$, $\lambda_i$. This system has a unique solution:
\begin{eqnarray}
\fl \det{ \left( \matrix{ 1 & -r_{\lone} \cr
1 & -r_{\ltwo}} \right)}=r_{\lone}-r_{\ltwo}= \frac{C_{\lone}}{D-C_{\lone}}-\frac{C_{\ltwo}}{D-C_{\ltwo}}=\frac{C_{\lone}(D-C_{\ltwo})-C_{\ltwo}(D-C_{\lone})}{(D-C_{\lone})(D-C_{\ltwo})} \nonumber\\
=\frac{D(C_{\lone}-C_{\ltwo})}{(D-C_{\lone})(D-C_{\ltwo})} \neq 0.
\end{eqnarray}

With \eref{eqn:r1r2}, \eref{eqn:cLl}, \eref{eqn:p1p2}, \eref{eqn:fanbeammain1} we can analytically compute all geometric calibration parameters $\lambda_i$, $y_{\lambda_i}$, $i=1,\ldots,P-1$ ($\lambda_0$ and $y_{\lambda_0}$ are chosen to be $0$) and the locations $C_{\lone}$, $p_{\lone}$ and $C_{\ltwo}$, $p_{\ltwo}$ of both 4-marker sets from $v^l_{ij}$, $j=1,\ldots,4$, $i=0,\ldots,P-1$, $l\in\{\lone,\ltwo\}$, projections of the markers in the projection data.

\subsection{Numerical simulations}
\label{sect:NumericalSimulationFanBeam}
The calibration marker sets are presented in~\Fref{fig:truncationfanbeam}. All distances are in cm. 
The $x_2$-coordinates of each 4-marker calibration object are respectively $p_{\lone}-k_1L$, $p_{\lone}-L$, $p_{\lone}+L$, $p_{\lone}+k_1L$  for the first line (where $L=0.4$ and $k_1=3$ are known), $p_{\ltwo}-k_3L$, $p_{\ltwo}-k_2L$, $p_{\ltwo}+k_2L$ and $p_{\ltwo}+k_3L$ for the second line (where $k_2=1$ and $k_3=2$ are known). The cross-ratio for the first line is $\frac{2L \cdot 2L}{4L \cdot 4L}=\frac{1}{4}$, for the second line is $\frac{L \cdot L}{3L \cdot 3L}=\frac{1}{9}$ (according to the definition from~\cite{Hartley2000} $ \textnormal{Cross}(z_1,z_2,z_3,z_4)=\frac{|z_1 z_2|\cdot|z_3z_4|}{|z_1 z_3|\cdot|z_2z_4|}$ with the distance $|z_i z_j|$ between two points $z_i$ and $z_j$).
The known distance $D$ between the source and detector lines is $10$. The unknown parameters are:
\begin{itemize}
    \item the positions of both 4-marker calibration objects:  $p_{\lone}=0$ and the distance to the detector line $C_{\lone}=1.5$ (the $x_1$-coordinate of 4 markers), $p_{\ltwo}=3.2$ and the distance to the detector line $C_{\ltwo}=0.5$ respectively;
    \item the source positions $\lambda_i$, $i=0,\ldots,P-1$, $P=30$,  are random values uniformly distributed on the interval $[-5,5]$;
    \item the detector jitters $y_{\lambda_i}$, $i=0,\ldots,P-1$, are random values uniformly distributed on the interval $[-0.05,0.05]$.
\end{itemize}

\begin{table}[bthp]
\caption{\label{tabfanbeam} Mean absolute errors for calibration parameters and the positions of the markers for numerical experiments with non-noisy and noisy projections. All errors are in cm and defined in~\eref{eqn:lambda_err} and~\eref{eqn:C_err}.} 

\begin{indented}
\lineup
\item[]\begin{tabular}{@{}*{6}{l}}
\br                              
Noise level&Noise std&Error for $\lambda_i$ &Error for $y_{\lambda_i}$  &Error for $p_l$ &Error for $C_l$ \cr 
~& ~ & $\ErrLambda$ & $\ErrY$ & $\ErrP$ & $\ErrC$ \cr
\mr
0\%&0&$2.10\times10^{-14}$&$3.62\times10^{-15}$&$1.17\times10^{-15}$&$3.39\times10^{-15}$\cr
    10\%&0.001&$2.32\times10^{-2}$&$3.40\times10^{-3}$&$1.26\times10^{-3}$&$4.58\times10^{-3}$\cr 
    50\% &0.005&$1.01\times10^{-1}$&$1.52\times10^{-2}$&$5.86\times10^{-3}$&$2.12\times10^{-2}$\cr 
    100\% &0.01&$2.20\times10^{-1}$&$3.13\times10^{-2}$&$1.15\times10^{-2}$&$4.32\times10^{-2}$\cr
    200\% &0.02&$5.19\times10^{-1}$&$7.45\times10^{-2}$&$2.63\times10^{-2}$&$9.66\times10^{-2}$\cr
\br
\end{tabular}
\end{indented}
\end{table}

We computed the exact projections of the marker centers from their coordinates in the initial coordinate system and produced as input two arrays $v^l_{ij}$, $i=0,\ldots,P-1$, $j=1,\ldots,4$, each of size $4P$, $l\in\{\lone,\ltwo\}$ (in each of both arrays, each of $P$ rows contains the coordinates $v^l_{ij}$, $j=1,\ldots,4$, of the projections of 4 markers). 
We simulated also noisy projection data of the markers. We modeled the detection error with a Gaussian noise $N(0,\sigma)$ added to $v^l_{ij}$, where $\sigma$ is the product of the noise level and the pixel size  $0.01$~cm of  the detector line
 (e.g., for the noise level of $200\%$ $\sigma=0.02$~cm).

The geometric calibration results are given such that $\lambda_0=0$ and $y_{\lambda_0}=0$. In order to  compare our parameter estimations with their true values, we need to subtract $\lambda_0$ from the true $\lambda_i$ and  $y_{\lambda_0}$ from the true $y_{\lambda_i}$, $i=0,\ldots,P-1$. From the discussion in the subsection~\ref{SectFanBeamNonUniqueSol} on the non-uniqueness of the solution we also need to perform a transformation on the estimations of $p_{\lone}$ and $p_{\ltwo}$ in order to compare them correctly with their true values.
The connection between the sheared and initial ordinates is $x^{{\mathrm{sheared}}}_{2}=x^{{\rm initial}}_{2}-(\lambda_0/D+y_{\lambda_0}/D)x^{{\rm initial}}_{1}+y_{\lambda_0}$, see~\eref{Eq:ShearAndTransForFunction} and~\eref{Eq:ShearAndTrans}. With this formula, we can compare the initial and estimated ordinates of the centers of mass $p_{\lone}$ and $p_{\ltwo}$ of each 4-marker group.

We present numerical results in~\Tref{tabfanbeam}. With noise-free projections our algorithm gives an almost exact solution (noise level of $0\%$ in \Tref{tabfanbeam}). We show mean absolute errors from $\nbk = 100$ realizations of the Gaussian noise for each noise level. For source positions and detector shifts we computed 
\begin{eqnarray}
\label{eqn:lambda_err}
\ErrLambda\eqdef \frac{1}{\nbk}\sum_{k=1}^{\nbk}
\frac{1}{P}\sum_{i=0}^{P-1} |\hat{\lambda}_{i,k} - \lambda_{i}|, \:
\ErrY\eqdef\frac{1}{\nbk}\sum_{k=1}^{\nbk}
\frac{1}{P}\sum_{i=0}^{P-1} |\hat{y}_{\lambda_i,k} - y_{\lambda_i}|,
\end{eqnarray}  
where $\lambda_{i}$ and $y_{\lambda_i}$ are the true values for the projection $i$, $\hat{\lambda}_{i,k}$ and $\hat{y}_{\lambda_i,k}$ are the corresponding estimations from our algorithm (see~\eref{eqn:fanbeammain1} and~\eref{eqn:r1r2}) of the source position and the detector shift for the projection $i$ with the Gaussian noise number $k$ added to projections $v^l_{ij}$. For the positions of both objects of four markers we computed 
\begin{eqnarray}
\label{eqn:C_err}
\ErrC\eqdef\frac{1}{\nbk}\sum_{k=1}^{\nbk} \frac{1}{2}\sum_{l\in\{\lone,\ltwo\}} |\hat{C}_{l,k} - C_{l}|, \:
\ErrP\eqdef \frac{1}{\nbk}\sum_{k=1}^{\nbk} \frac{1}{2}\sum_{l\in\{\lone,\ltwo\}} |\hat{p}_{l,k} - p_{l}|,
\end{eqnarray}
where $C_{l}$ and  $p_{l}$ are the true values, $\hat{C}_{l,k}$ and $\hat{p}_{l,k}$ are the estimations of the position of the marker objects, $l \in \{\lone,\ltwo\}$, from our algorithm (see~\eref{eqn:cLl},~\eref{eqn:p1p2} and~\eref{eqn:r1r2}) with the Gaussian noise number $k$ added to geometric projection data  $v^l_{ij}$.
With noisy projections, the errors are essentially proportional to the noise level. 
The results of our estimations of source positions are comparable with the estimations of Jonas~\cite{jonas18} for the general fan-beam geometry (Jonas estimated only source positions).

\section{Conclusion}
\label{SecConclusions}

Range conditions on functions in the so-called "projection form" (like HLCCs in~\cite{helgason65,ludwig66} or in~\cite{clackdoyle13} for the fan-beam linogram transform) have been applied to the geometric self-calibration of X-ray systems. In this work, we showed that range conditions on distributions can be considered as a powerful tool to solve the similar task. We constructed first DCCs on distributions for the fan-beam geometry and simplified existing DCCs for the parallel geometry. DCCs are based on moments of projections. Thus, DCCs on functions can not be computed if the projections are truncated. Singularities like Dirac distributions can model markers introduced in the X-ray system field-of-view.   If all singularities are in the field-of-view of projections, their projections aren't truncated. Thus, the range conditions of projections of singularities can be computed if full projections are truncated.

We have illustrated these ideas in two 2D projection geometries: the classical 2D parallel projection geometry (the 2D Radon transform) and the 2D fan-beam geometry with sources on a line. 
For both geometries we have proposed a {\em partially known} set of Dirac distributions and associated analytical formulas to identify geometrical projection parameters: the angles for the Radon transform or the source positions in fan-beam geometry with sources on a line and the detector jitters in both geometries. With the derived methods, we can analytically estimate a large number of geometric calibration parameters. We have presented numerical experiments where $2P$ geometrical parameters were estimated from $P$ projections, the angles and shifts with $P=80$ in parallel geometry and the source positions and detector shifts with $P=30$ in fan-beam linogram geometry respectively. From the derived DCCs on distributions, we observed that the behavior of moments remains polynomial in the case of projections of distributions. This demonstrates that the computation of moments on distributions is instrumental in deriving meaningful and easily manipulable equations, which can significantly aid in identifying the parameters for calibration.

We have numerically shown that our derived calibration methods are robust to noise. The error of the parameter estimation behaves essentially linearly relatively to the noise level in the range from $10\%$ to $200\%$ of the detector pixel size, see the subsections~\ref{sect:NumericalSimulationParallel} and~\ref{sect:NumericalSimulationFanBeam}. Contrary to many self-calibration methods based on DCCs on functions or bundle adjustment~\cite{basu100,basu200,lesaint2017a,aichert15,konik21,unberath17}, we do not use iterative methods (to solve non-linear problems), but we have closed-form solutions to estimate the geometric calibration parameters from moment conditions. Of course, we present only the initial attempt to exploit the idea of calibrating using moment conditions on distributions. One area for improvement, as it presented a disadvantage, is for the fan-beam geometry, where we need to place a pair of calibration marker sets parallel to the detector line. However, in our algorithm, their positions are flexible, which is an advantage, as it ensures the markers can be positioned within the X-ray system field-of-view, mitigating the potential limitation. The future work could include the application of the proposed idea of the utilization of moment conditions on distributions solely on non-truncated singularities not only for singularities represented by introduced markers with partially known geometries, but also for intrinsic singularities of the image, as used in computer vision. This approach would help avoid disturbing the scene with markers, but it would likely require the application of optimization techniques to solve the equations derived from the moment conditions on distributions.

\ack
This work is supported by the French National Research Agency in the framework of the "Investissement d'avenir" program (ANR-15-IDEX-02), the "Fonds unique interministériel", the European Union FEDER in Auvergne-Rhône-Alpes (3D4Carm) and  the ANR CAMI LabEx (ANR-11-LABX-0004-01).

\section{Appendix}
\label{SecAppendix}

In this section, we discuss the fact that the Radon transform at fixed $\alpha$ of distributions and that the fan-beam transform with sources on a line at fixed $\lambda$ (i.e., at fixed source position) of distributions are distributions.

A distribution $T$ acting on a set of functions is a continuous linear functional from this set to $\mathbb{R}$. The check of continuity of distributions can be replaced by the check of their boundedness, see~\cite{bony01}. 
$\mathscr{E}(\Omega_N)$ ($N=1$ or $N=2$ in our case) is usually considered as a Fréchet space, a special topological vector space equipped with a set of semi-norms.
\begin{definition}[\cite{bony01}, C.1, p.239]
For the vector space $X$ a non-negative real-value function $P:X \to [0,+\infty)$ is a semi-norm if it satisfies $P(f+g) \leq P(f)+P(g)$ and $P(\lambda f)=|\lambda|P(f)$ $\forall f\in  X$, $\forall g \in  X$ and $\forall\lambda\in\bR$.
\end{definition}
Semi-norms help to define a proper topology in $\mathscr{E}(\Omega_N)$. According to~\cite{bony01}, Chapter C.2.6, p.243 and Chapter 2.3.B, p.38, the set of semi-norms is defined 
\begin{itemize}
    \item in the case $N=2$ by
    \begin{eqnarray}
    \Ptwoi(\psi)=\sum_{|\bbeta| \leq i } \sup_{\vec{x} \in \Ktwoi} \abs{\partial^{\bbeta} \psi (\vec{x})},\: i\in\bN,
    \label{eqn:seminormsbasictwo}
    \end{eqnarray}
    where $\psi \in \mathscr{E}(\Omega_2)$,  $\bbeta=(\beta_1,\beta_2) \in \mathbb{N}^2$, $|\bbeta|=\beta_1+\beta_2$, and $\partial^{\bbeta} = \partial_1^{\beta_1}\partial_2^{\beta_2}$ are partial derivatives of order $\beta_k$ according to $x_k$, $k=1,2$, 
    the set $\{\Ktwoi\}_{i\in\bN}$ is a set of compacts with $\cup_{i\in\bN} \Ktwoi = \Omega_2$, $\Ktwoi$ is a subset of the interior of the set $\Ktwoipo$;
\item in the case $N=1$ by 
    \begin{eqnarray}
    \Ponej(\phi)=\sum_{\gamma \leq j } \sup_{s \in \Konej} \abs{\phi^{(\gamma)} (s)}, \: j\in\bN,
    \label{eqn:seminormsbasicone}
    \end{eqnarray}
    where $\phi \in \mathscr{E}(\Omega_1)$, $\gamma \in \mathbb{N}$, 
    the set $\{\Konej\}_{j\in\bN}$ is a set of compacts with $\cup_{j\in\bN} \Konej = \Omega_1$, $\Konej$ is a subset of the interior of the set $\Konejpo$.
\end{itemize}

According to~\cite{bony01}, Theorem C.1.7, Chapter C.1, p.240, for $T: \mathscr{E} (\Omega_N)\to \mathbb{R}$, the continuity of the functional $T$ is equivalent to its boundedness. The functional $T$ is bounded if 
\begin{eqnarray}
\exists j \in \bN, \: \exists C>0, \: \forall \phi \in \mathscr{E} (\Omega_N)  \: |T(\phi)| \leq C P_j(\phi), 
\label{eqn:boundedT}
\end{eqnarray}
where $P_j$ is a semi-norm on $\mathscr{E} (\Omega_N)$ and we used the remark that the normed space $\left(\bR,\abs{\cdot}\right)$ is a particular case of a semi-normed space, see Chapter C.1, p.239,~\cite{bony01}.

\begin{lemma}\label{LemmaRadonBounded}
$\mathcal{R}_{\alpha}f$ defined in \eref{eqn:dualityradonour} is bounded.
\end{lemma}

\begin{proof}
We have seen that $\forall \phi\in \mathscr{E}_1$, $\mathcal{R}_{\alpha}^{*}(\phi) \in \mathscr{E}_2$. Since $f$ is bounded, according to~\eref{eqn:boundedT}, $\exists I{\color{black} \in \bN}$, $\exists C_1 {\color{black} > 0}$, $\forall \phi\in \mathscr{E}_1$   
\begin{eqnarray}
|\mathcal{R}_{\alpha}f(\phi)|=|f\left( \mathcal{R}_{\alpha}^{*}(\phi)\right| \leq C_1 \PtwoI\left( \mathcal{R}_{\alpha}^{*}(\phi)\right),
\label{eqn:bounded_f_of_RalphaDiese}
\end{eqnarray}
where  $\Ptwoi$, $i \in \mathbb{N}$, are semi-norms on $\mathscr{E}_2$.

To prove that $\mathcal{R}_{\alpha}f$ is bounded, we need to show that 
$\exists J \in \bN$, $\exists C_2 > 0$, $\forall \phi \in \mathscr{E}_1$
\begin{eqnarray}
|\mathcal{R}_{\alpha}f(\phi)|\leq  C_2 \PoneJ(\phi),
\end{eqnarray}
where $\Ponej$, $j \in \mathbb{N}$, are semi-norms on $\mathscr{E}_1$. From~\eref{eqn:bounded_f_of_RalphaDiese}, it's sufficient to show that $\forall I \in \bN$, $\exists J \in \bN$, $\exists C_3 > 0$, $\forall \phi \in \mathscr{E}_1$
\begin{eqnarray}
\PtwoI\left( \mathcal{R}_{\alpha}^{*}(\phi)\right) \leq  C_3 \PoneJ(\phi).
\label{eqn:seminormsradon}
\end{eqnarray}

Each $\Ptwoi$, $i\in\bN$, in~\eref{eqn:seminormsbasictwo} is a finite sum over $\bbeta$ of expressions
\begin{eqnarray}
\Ptwotibbeta(\psi)=\sup_{\vec{x} \in \Ktwoi} |\partial^{\bbeta} \psi (\vec{x})|,
\label{eqn:Ptwotibbeta}
\end{eqnarray}
where $\Ktwoi$ is a compact set, $\cup_{i \in \bN} \Ktwoi=\mathbb{R}^2$, $\Ktwoi$ is in the interior of the set $\Ktwoipo$, $\psi\in\mathscr{E}_2$ (in particular, when  $\psi(\vec{x})=\mathcal{R}_{\alpha}^{*}(\phi)(\vec{x}) =\phi(\vec{x} \cdot \vec{\theta}_{\alpha})$). Each $\Ponej$,  $j\in\bN$, in~\eref{eqn:seminormsbasicone} is a finite sum over $\gamma$ of expressions
\begin{eqnarray}
\Ponetjgamma(\phi)=\sup_{s \in \Konej} \abs{\phi^{(\gamma)} (s)},
\label{eqn:Ponetjgamma}
\end{eqnarray}
where $\Konej$ is a compact set, $\cup_{j \in \bN} \Konej=\mathbb{R}$, $\Konej$ is in the interior of the set $\Konejpo$.
Thus, it's sufficient to show that $\forall \bbeta$, $\exists \gamma(\bbeta)  \in \bN$, $\exists j(\bbeta) \in \bN$, $\exists C(\bbeta)>0$, $\forall \phi \in \mathscr{E}_1$
\begin{eqnarray}
\PtwotIbbeta\left(\mathcal{R}_{\alpha}^{*}(\phi)\right) \leq C(\bbeta) \PonetJgammabeta(\phi).
\label{eqn:seminormstoproveradon}
\end{eqnarray}

$\PtwoI$ is a finite sum over $\bbeta$ of $\PtwotIbbeta$. Thus, if we have \eref{eqn:seminormstoproveradon}, we can define one common constant $C_3$ as the maximum of all $C(\bbeta)$, one common $J$ such that it corresponds to sufficiently large compact $\KoneJ$ containing all compacts $\Konejbbeta$ and fulfilling $\gamma (\bbeta) \leq J$ $\forall \gamma (\bbeta)$, then \eref{eqn:seminormstoproveradon} gives \eref{eqn:seminormsradon}.

Let us show \eref{eqn:seminormstoproveradon}. For each function $\phi \in \mathscr{E}_1$: $\psi(\vec{x})\eqdef\phi(\vec{x} \cdot \vec{\theta}_{\alpha})=\phi(x_1 \cos{\alpha}+x_2 \sin{\alpha})$ we obtain $\partial^{\beta_1}_{1}\partial^{\beta_2}_{2}\psi(\vec{x}) =(\cos{\alpha})^{\beta_1} (\sin{\alpha})^{\beta_2} \phi^{(|\bbeta|)} (\vec{x} \cdot \vec{\theta}_{\alpha})=(\cos{\alpha})^{\beta_1} (\sin{\alpha})^{\beta_2} \phi^{(\gamma (\bbeta))} (\vec{x} \cdot \vec{\theta}_{\alpha})$ with $\bbeta=(\beta_1,\beta_2)\in \mathbb{N}^2$, $\gamma (\bbeta)\eqdef|\bbeta|$. 
Since $\abs{\vec{x} \cdot \vec{\theta}_{\alpha}}\leq \norm{\vec{x}}$  and $\vec{x}$ is in the compact set $\KtwoI$, because $\cup_{j \in \bN} \Konej=\mathbb{R}$ with $ \Konej\subset  \Konejpo$, 
for the compact set $\KtwoI$, we can choose $j(\bbeta)$ (depending also on $I$) sufficiently large such that $\vec{x} \cdot \vec{\theta}_{\alpha}=x_1 \cos{\alpha}+x_2 \sin{\alpha}$ is in $\KonejbbetaI$ for all $\vec{x} \in \KtwoI$ with $\gamma (\bbeta) \leq j(\bbeta)$. Thus, 
\begin{eqnarray}
\sup_{\vec{x} \in \KtwoI} \abs{\phi^{(\gamma (\bbeta)) } (\vec{x} \cdot \vec{\theta}_{\alpha})} \leq \sup_{s \in \KonejbbetaI} \abs{\phi^{(\gamma (\bbeta)) } (s)}.   
\end{eqnarray} 
We can conclude
\begin{eqnarray}
\fl \PtwotIbbeta(\psi)=\sup_{\vec{x} \in \KtwoI} \abs{\partial^{\bbeta} \psi (\vec{x})} \leq  C(\bbeta)\sup_{\vec{x} \in \KtwoI} \abs{\phi^{(\gamma (\bbeta))} (\vec{x} \cdot \vec{\theta}_{\alpha})}  \leq C(\bbeta) \sup_{s \in \KonejbbetaI} \abs{\phi^{(\gamma (\bbeta))} (s)} \nonumber \\ = C(\bbeta) \PonetJgammabeta(\phi),
\end{eqnarray} 
where $C(\bbeta) =|(\cos{\alpha})^{\beta_1} (\sin{\alpha})^{\beta_2}|$. This yields the proof.
\end{proof}

\begin{lemma}\label{LemmaFanbeamBounded}
$\mathcal{D}_{\lambda}f$ defined in \eref{eqn:dualityfanbeam} is bounded.
\end{lemma}

\begin{proof}
By analogy with the previous proof, we define the set of compacts $\Konej$, where $\cup_{j \in \bN} \Konej=\mathbb{R}$, $\Konej$ is in the interior of the set $\Konejpo$, each semi-norm $\Ponej$ of $\mathscr{E}_1$ is a finite sum over $\gamma$ of $\Ponetjgamma$ as in~\eref{eqn:Ponetjgamma}. 
We consider $\Ktwoi=[d^i_1,d_2^i]\times[-b^i,b^i]$ with $\cup_{i \in \bN} \Ktwoi=X_2$, where $X_2=(D_1,D_2) \times \mathbb{R}$, $0<D_1<D_2<D$ with $0<d^i_1<d^i_2<D$, $b^i >0$, $b^i \in \mathbb{R}$, 
$0<D_1<d^{i+1}_1 < d^{i}_1 < d^{i}_2 < d^{i+1}_2<D_2<D$, $0<b^{i} < b^{i+1}$, each semi-norm $\Ptwoi$ of $\mathscr{E}(X_2)$ is a finite sum over $\bbeta$ of $\Ptwotibbeta$, as in~\eref{eqn:Ptwotibbeta}.

We have seen that $\forall \phi\in \mathscr{E}_1$, $\mathcal{D}_{\lambda}^{*}(\phi) \in \mathscr{E}(X_2)$. Then since $f$ is bounded, according to~\eref{eqn:boundedT}, $\exists I{\color{black} \in \bN}$, $\exists C_1 {\color{black} > 0}$, $\forall \phi\in \mathscr{E}_1$   
\begin{eqnarray}
|\mathcal{D}_{\lambda}f(\phi)|=\left|f\left( \mathcal{D}_{\lambda}^{*}(\phi)\right)\right| \leq C_1 \PtwoI\left( \mathcal{D}_{\lambda}^{*}(\phi)\right).
\label{eqn:bounded_fDstar}
\end{eqnarray}

To prove that $\mathcal{D}_{\lambda}f$ is bounded, we need to show that 
$\exists J \in \bN$, $\exists C_2 > 0$, $\forall \phi \in \mathscr{E}_1$
\begin{eqnarray}
|\mathcal{D}_{\lambda}f(\phi)|\leq  C_2 \PoneJ(\phi).
\end{eqnarray}
Thus, from~\eref{eqn:bounded_fDstar}, it's sufficient to show that $\forall I\in\bN$, $\exists J \in \bN$, $\exists C_3 > 0$, $\forall \phi \in \mathscr{E}_1$
\begin{eqnarray}
\PtwoI\left( \mathcal{D}_{\lambda}^{*}(\phi)\right) \leq  C_3 \PoneJ(\phi).
\label{eqn:seminormsfanbeam}
\end{eqnarray}

For that, it's sufficient to show that $\forall \bbeta\in\bN^2$,  there is a finite set  $\GammaBB\subset\bN$ (of $\gammaBd$), $\exists j(\bbeta) \in \bN$, $\exists C(\bbeta) > 0$, $\forall \phi \in \mathscr{E}_1$
\begin{eqnarray}
\PtwotIbbeta \left( \mathcal{D}_{\lambda}^{*}(\phi)\right)\leq C(\bbeta) \sum_{\gammaBd \in \GammaBB} \PonetJgamma(\phi),
\label{eqn:seminormstoprovefanbeam}
\end{eqnarray}
where the sum over $\gammaBd$ in the expression is finite.  
 Similarly with the previous proof, if we have \eref{eqn:seminormstoprovefanbeam}, we can define one common constant $C_3$ as the maximum of all $C(\bbeta)$, one common $J$ such that it corresponds to sufficiently large compact $\KoneJ$ containing all compacts $\Konejbbeta$ and fulfilling $\gammaBd \leq J$ for all sequences of $\gammaBd$, then \eref{eqn:seminormstoprovefanbeam} gives \eref{eqn:seminormsfanbeam}, because $\PtwoI$ is a finite sum of $\PtwotIbbeta$.

From now, we define $\psi(\vec{x})\eqdef \mathcal{D}_{\lambda}^{*}(\phi)(\vec{x}) = \frac{1}{D-x_1}\phi \left( \frac{x_2D-x_1\lambda}{D-x_1} \right)=\frac{1}{D-x_1}\phi \left( \lambda+ \frac{D\left(x_2-\lambda\right)}{D-x_1}  \right)$.

Let us start with the simple case $\bbeta=(0,0)$: 
\begin{eqnarray}
\fl \tilde{P}_{2,I,(0,0)}(\psi)=\sup_{\vec{x} \in \KtwoI} |\psi (\vec{x})|=\sup_{\vec{x} \in \KtwoI} \abs{\frac{1}{D-x_1}  \phi \left( \frac{x_2D-x_1\lambda}{D-x_1} \right) } \nonumber \\ \leq \frac{1}{D-D_2} \sup_{\vec{x} \in \KtwoI} \abs{  \phi \left( \frac{x_2D-x_1\lambda}{D-x_1} \right) } .
\end{eqnarray}
Now
\begin{eqnarray}
\fl \abs{ \frac{x_2D-x_1\lambda}{D-x_1} } = \abs{ \lambda+ D \frac{x_2-\lambda}{D-x_1} } \leq \abs{ \lambda}+D \abs{  \frac{x_2-\lambda}{D-x_1} } \leq \abs{ \lambda} + D \frac{|x_2|+|\lambda|}{D-D_2}  \nonumber \\ \leq |\lambda| + D \frac{b^I+|\lambda|}{D-D_2},
\end{eqnarray}
thus, we can find $j(\bbeta)$ (for $\bbeta=(0,0)$, depending on $I$) and $\Konejbbeta$ such that $\left\{s\in\bR, |s|\leq |\lambda| + D \frac{b^I+|\lambda|}{D-D_2}\right\} \subset \Konejbbeta$ ($\lambda$ is fixed) to obtain
\begin{eqnarray}
\tilde{P}_{2,I,(0,0)}(\psi)\leq \frac{1}{D-D_2} \sup_{s \in \Konejbbeta} |\phi(s)|
= C(\bbeta) \tilde{P}_{1,j(\bbeta),0}(\phi).
\end{eqnarray}

We want to show by induction that the function inside each $\PtwotIbbeta(\psi)$ (defined as in~\eref{eqn:Ptwotibbeta}) is a finite sum of expressions 
\begin{eqnarray}
\psi_{\gammaBd,p_1,p_2}(\vec{x})
=\phi^{(\gammaBd)} \left( \lambda+ D \frac{x_2-\lambda}{D-x_1}  \right) 
\frac{C_{\gammaBd,p_1,p_2}(x_2-\lambda)^{p_2}}{(D-x_1)^{p_1}},
    \label{eqn:expressiontypefanbeam}
\end{eqnarray}
where $0\leq\gammaBd\leq\beta_1+\beta_2$, $1\leq p_1\leq 2\beta_1 + \beta_2 + 1$ and $0 \leq p_2 \leq \beta_1$.

It's obvious for $\bbeta=(0,0)$ with $p_1=1$ and $p_2=0$. Then we consider:
\begin{itemize}
\item The case $\bbeta=(1,0)$:
\begin{eqnarray}
\fl \partial^{1}_1 \psi(\vec{x})
= \phi' \left( \lambda+ D \frac{x_2-\lambda}{D-x_1}  \right) \frac{D(x_2-\lambda)}{(D-x_1)^3}+\phi \left( \lambda+ D \frac{x_2-\lambda}{D-x_1}  \right) \frac{1}{(D-x_1)^2},
\end{eqnarray}
here $\gamma\leq 1(\leq\beta_1+\beta_2 =1)$, $p_1\leq 3 (\leq 2\beta_1+\beta_2+1 =3)$, $p_2 \leq 1 (\leq \beta_1=1)$.

\item The case $\bbeta=(0,1)$:
\begin{eqnarray}
    \partial^{1}_2 \psi(\vec{x})= \phi' \left( \lambda+ D \frac{x_2-\lambda}{D-x_1}  \right) \frac{D}{(D-x_1)^2},
\end{eqnarray}
here $\gamma = 1(\leq\beta_1+\beta_2 =1)$, $p_1=2 (\leq 2\beta_1+\beta_2+1 =2)$, $p_2 = 0(\leq \beta_1=0)$.
\end{itemize}

The induction step:
\begin{itemize}

\item We differentiate \eref{eqn:expressiontypefanbeam} according to the first variable and we again have the finite sum of expressions of the type \eref{eqn:expressiontypefanbeam}:
\begin{eqnarray}
\fl \partial^{1}_1 \psi_{\gammaBd,p_1,p_2}(\vec{x})
=\phi^{(\gammaBd+1)} \left( \lambda+ D \frac{x_2-\lambda}{D-x_1}  \right) 
{D} \frac{C_{\gammaBd,p_1,p_2}(x_2-\lambda)^{p_2+1}}{(D-x_1)^{p_1 + 2}} \nonumber \\
+\phi^{(\gammaBd)} \left( \lambda+ D \frac{x_2-\lambda}{D-x_1}  \right) \frac{C_{\gammaBd,p_1,p_2}p_1(x_2-\lambda)^{p_2}}{(D-x_1)^{p_1+1}}. 
\end{eqnarray}
If $\gammaBd\leq\beta_1+\beta_2$, then $\gammaBd<\gammaBd+1 \leq (\beta_1+1)+\beta_2$; 
if  $1\leq p_1\leq 2\beta_1 + \beta_2 + 1$, then $1\leq p_1+1< p_1+2 \leq 2(\beta_1+1) + \beta_2 + 1$; 
if $0 \leq p_2 \leq \beta_1$, then  $0 \leq p_2 < p_2+1 \leq \beta_1+1$.

\item We differentiate \eref{eqn:expressiontypefanbeam} according to the second variable and we again have the finite sum of expressions of the type \eref{eqn:expressiontypefanbeam}:
\begin{eqnarray}
\fl \partial^{1}_2 \psi_{\gammaBd,p_1,p_2}(\vec{x})= \phi^{(\gammaBd+1)} \left( \lambda+ D \frac{x_2-\lambda}{D-x_1}  \right) \frac{DC_{\gammaBd,p_1,p_2}(x_2-\lambda)^{p_2}}{(D-x_1)^{p_1+1}} \nonumber \\+\phi^{(\gammaBd)} \left( \lambda+ D \frac{x_2-\lambda}{D-x_1}  \right) \frac{C_{\gammaBd,p_1,p_2}p_2(x_2-\lambda)^{p_2-1}}{(D-x_1)^{p_1}}.
\end{eqnarray}
If $\gammaBd\leq\beta_1+\beta_2$, then $\gammaBd<\gammaBd+1 \leq \beta_1+(\beta_2+1)$; 
if  $1\leq p_1\leq 2\beta_1 + \beta_2 + 1$, then $1\leq p_1< p_1+1 \leq 2\beta_1 + (\beta_2+1) + 1$; 
if $0 \leq p_2 \leq \beta_1$, then  $0 \leq p_2-1 < p_2 \leq \beta_1$.

\end{itemize}

For $\KtwoI$ we can find $\Konejbbeta$ with  all $\gammaBd \leq j(\bbeta)$ such that
\begin{eqnarray}
\fl \PtwotIbbeta (\psi) =  \sup_{\vec{x} \in \KtwoI} \abs{ \sum_{\gammaBd,p_1,p_2} \psi_{\gammaBd,p_1,p_2} (\vec{x})  }\nonumber \\
= \sup_{\vec{x} \in \KtwoI} \abs{ \sum_{\gammaBd,p_1,p_2} \phi^{(\gammaBd)} \left( \lambda+ D \frac{x_2-\lambda}{D-x_1}  \right) \frac{C_{\gammaBd,p_1,p_2}(x_2-\lambda)^{p_2}}{(D-x_1)^{p_1}} } \nonumber \\ 
\leq \sum_{\gammaBd,p_1,p_2}  \sup_{\vec{x} \in \KtwoI} \abs{ \phi^{(\gammaBd)} \left( \lambda+ D \frac{x_2-\lambda}{D-x_1}  \right) \frac{C_{\gammaBd,p_1,p_2}(x_2-\lambda)^{p_2}}{(D-x_1)^{p_1}}  } \nonumber \\  
\leq \sum_{\gammaBd,p_1,p_2} \sup_{s \in \Konejbbeta} \abs{\phi^{(\gammaBd)} (s)} \frac{C_{\gammaBd,p_1,p_2}(b^I+|\lambda|)^{p_2}}{(D-D_2)^{p_1}}  \nonumber \\ 
\leq C(\bbeta) \sum_{\gammaBd} \sup_{s \in \Konejbbeta} |\phi^{(\gammaBd)} (s)| = C(\bbeta) \sum_{\gammaBd} \PonetJgamma(\phi). 
\end{eqnarray}
This yields the proof.
\end{proof}

\section*{References}
\bibliography{myref}

\end{document}